\def\doccolumn{2}

\if\doccolumn2
    \documentclass[journal]{IEEEtran}
\else
    \documentclass[journal,draftclsnofoot,onecolumn,12pt,twoside]{IEEEtran}
\fi

\usepackage{times}
\usepackage{bbm}
\usepackage{epsfig}
\usepackage{color}
\usepackage{amsthm}
\usepackage{amsmath}
\usepackage{amssymb}
\usepackage{multirow}
\usepackage{stfloats}
\usepackage{booktabs}
\usepackage{algorithm}
\usepackage{algorithmic}
\usepackage{color}
\usepackage{cite}
\usepackage{graphics}
\usepackage{graphicx}
\usepackage{dsfont}
\usepackage{mathtools}
\usepackage{stmaryrd}
\usepackage{mathrsfs}
\usepackage{nccmath}
\usepackage{cuted}
\usepackage{tabularx}
\usepackage{mleftright}
\usepackage{cases}
\usepackage{lineno}
\usepackage{dsfont}
\usepackage{lipsum}
\usepackage{xpatch}
\usepackage[ddmmyyyy]{datetime}
\PassOptionsToPackage{hyphens}{url}\usepackage{hyperref}

\if \doccolumn1
\fi

\makeatletter
\newcommand{\mathleft}{\@fleqntrue\@mathmargin0pt}
\newcommand{\mathcenter}{\@fleqnfalse}
\makeatother

\ifCLASSOPTIONcompsoc
    \usepackage[caption=false, font=normalsize, labelfont=sf, textfont=sf]{subfig}
\else
    \usepackage[caption=false, font=footnotesize]{subfig}
\fi

\makeatletter
\DeclareRobustCommand\bfseriesitshape{%
  \not@math@alphabet\itshapebfseries\relax
  \fontseries\bfdefault
  \fontshape\itdefault
  \selectfont
}

\makeatother
\DeclareTextFontCommand{\textbfit}{\bfseriesitshape}

\makeatletter
\xpatchcmd\proof{\@addpunct{.}}{\@addpunct{:}}{}{}
\xpatchcmd\proof{\itshape}{\itshape}{}{}
\xpatchcmd\proof{\hskip}{\hskip3}{}{}
\makeatother

\if\doccolumn2
\newtheoremstyle{newthmstyle}{\topsep}{\topsep}{}{\parindent}{\itshape}{:}{ }{}
\fi

\theoremstyle{newthmstyle}
\newtheorem{thm}{Theorem}
\newtheorem{define}[thm]{Definition}
\newtheorem{lem}[thm]{Lemma}

\newtheorem{cor}[thm]{Corollary}

\begin{document}

\title{Rate Matching for Polar Codes \\ Based on Binary Domination}

\author{Min~Jang,~\IEEEmembership{Member,~IEEE},
        Seok-Ki~Ahn,
        Hongsil~Jeong,
        Kyung-Joong~Kim,
        Seho~Myung,
        Sang-Hyo~Kim,~\IEEEmembership{Member,~IEEE},
        and Kyeongcheol~Yang,~\IEEEmembership{Senior Member,~IEEE}
\\
\if\doccolumn1
Submitted October 22, 2018
\vspace{-1.8cm}
\fi
\thanks{This research was supported in part by the National Research Foundation (NRF) of Korea
under grant funded by the Ministry of Science and Information \& Communication Technology (MSIT) of the Korea Government
(2016R1A2A1A05005023)
and in part by the Institute for Information \& communications
Technology Promotion (IITP) under grant funded by the MSIT of the Korea government
(2018(2016-0-00123)).}
\thanks{M. Jang, H. Jeong, K.-J. Kim, and S.~Myung are with Samsung Electronics, Suwon, Gyeonggi 16677, South Korea (e-mail: \{mn.jang, hongsil.jeong, kj1981.kim, seho.myung\}@samsung.com).}
\thanks{S.-K. Ahn is with Electronics and Telecommunications Research Institute (ETRI), Daejeon 34129, South Korea (e-mail: seokki.ahn@etri.re.kr).}
\thanks{S.-H. Kim is with the College of Information and Communication Engineering, Sungkyunkwan University, Suwon, Gyeonggi 16419, South Korea (e-mail: iamshkim@skku.edu).}
\thanks{K. Yang is with the Dept. of Electrical Engineering, Pohang University of Science and Technology (POSTECH), Pohang, Gyeongbuk 37673, South Korea (e-mail: kcyang@postech.ac.kr).}
}


\maketitle

\thispagestyle{empty}
\begin{abstract}
In this paper, we investigate the fundamentals of puncturing and shortening for polar codes, based on  \textit{binary domination} which plays a key role in polar code construction.
We first prove that the orders of encoder input bits to be made incapable (by puncturing) or to be shortened are governed by binary domination.
In particular, we show that binary domination completely determines incapable or shortened bit patterns for polar codes,
and that all the possible incapable or shortened bit patterns can be identified.
We then present the patterns of the corresponding encoder output bits to be punctured or fixed, when the incapable or shortened bits are given.
We also demonstrate that the order and the pattern of puncturing and shortening for polar codes can be aligned.
In the previous work on the rate matching for polar codes,
puncturing of encoder output bits begins from a low-indexed bit, while shortening starts from a high-indexed bit.
Unlike such a conventional approach, we show that encoder output bits can be punctured from high-indexed bits, while keeping the incapable bit pattern exactly the same.
This makes it possible to design a unified circular-buffer rate matching (CB-RM) scheme that includes puncturing, shortening, and repetition.
\end{abstract}

\if\doccolumn1
\vspace{-0.6cm}
\fi
\begin{IEEEkeywords}
Polar codes, code modification, rate matching, binary domination, puncturing, shortening.
\end{IEEEkeywords}

\IEEEpeerreviewmaketitle

\section{Introduction} \label{sec:introduction}

Polar codes, proposed by Ar{\i}kan in \cite{Arikan2009}, are a class of error-correcting codes first proved to achieve the symmetric capacity of an arbitrary binary-input discrete memoryless channel (B-DMC) under low-complexity decoding.
It was also shown in \cite{Tal2015,Niu2012} that polar codes achieve practically good finite-length performance under successive-cancellation list (SCL) decoding when they are concatenated with an outer code.
For this reason, the 3rd Generation Partnership Project (3GPP) recently agreed to adopt polar codes for control information in the 5G New Radio (NR) access technology \cite{TS38212}.

In the construction of a polar code, the $2\times 2$ polarization kernel matrix is generally considered,
and the length is fixed to a power of two.
In order to make the length arbitrary for practical applications, rate-matching schemes such as puncturing, shortening, and repetition are applied to the encoder output.
They change the effective bitwise channels that the encoder output bits experience,
and thereby alter overall channel polarization.
Therefore, the code construction and rate matching for polar codes need to be simultaneously taken into account to achieve good rate-compatible performance.

Numerous rate-matching schemes for polar codes have been widely studied in the literature \cite{Eslami2011,Niu2013a,Shin2013,Zhang2014,Kim2015,Saber2015,Chandesris2017,Hong2017,Hong2018,El-Khamy2018,Wang2014,Miloslavskaya2015,Bioglio2017}.
Among them, puncturing and shortening have been commonly studied.
First, puncturing is a classical modification method that reduces the length of a code, while maintaining its dimension.
In polar coding, puncturing coded bits changes overall channel polarization considerably, so that some polarized split channels become incapable of delivering information.
Practical puncturing schemes such as random and stopping-tree puncturing for polar codes were firstly studied in \cite{Eslami2011}.
A quasi-uniform puncturing scheme was proposed in \cite{Niu2013a} to design punctured polar codes with a good minimum row-weight property,
while a method to construct length-compatible polar codes by reducing polarizing matrices was studied in \cite{Shin2013}.
Especially in \cite{Shin2013}, an important term, \textit{incapable bits} resulting from puncturing, was first introduced.
Search algorithms by density evolution or Gaussian approximation were introduced in \cite{Zhang2014,Kim2015} to design an optimal puncturing pattern.
An interesting rate-matching scheme via puncturing and extending intermediate coded bits was investigated for incremental-redundancy hybrid auotmatic-repeat-and-request (IR-HARQ) in \cite{Saber2015}.
In \cite{Chandesris2017}, a class of symmetric puncturing patterns is introduced, and a method to efficiently generate symmetric puncturing patterns is proposed.
More recently, Hong \textit{et al.} \cite{Hong2017,Hong2018} showed the existence of capacity-achieving punctured polar codes,
and El-Khamy \textit{et al.} \cite{El-Khamy2018} designed a circular-buffer rate matching (CB-RM) based on two-stage polarization.

Shortening is another approach to modifying polar codes for rate matching.
In shortened polar codes, the values of some encoder output bits are fixed to a deterministic value, typically zero, by shortening some encoder input bits.
Although these fixed encoder output bits are not transmitted,
the decoder is aware of their values and is able to use the information for decoding.
A shortened polar code can be seen as a subcode of a given mother code.
Wang and Liu \cite{Wang2014} proposed a general way of constructing shortened submatrices by recursively eliminating a single-weight column in the mother matrix.
Recently, an efficient algorithm for finding good shortening bit patterns was proposed in \cite{Miloslavskaya2015},
and a rate-matching scheme designed to combine both puncturing and shortening was investigated in \cite{Bioglio2017}.

As a practical application, the polar coding scheme in 3GPP NR \cite{TS38212} exploits a single nested rate-matching pattern based on a subblock-wise permutation, which is commonly applied to puncturing, shortening, and repetition.
The encoder output bit sequence is divided into 32 subblocks,
and they are subblock-wise interleaved in a predetermined pattern, regardless of the rate-matching technique employed.
The interleaved bit sequence is then stored in a circular buffer, and as many bits as desired are selected from the buffer for transmission.
Note that the starting point for extracting the desired bits from the buffer is set differently, depending on the applied rate-matching scheme.
When either shortening or repetition is configured, the bits from the front of the buffer are selected for transmission.
On the other hand, when puncturing is set, the starting point is determined so that the bits located at the head of the buffer are not transmitted.
This CB-RM in conjunction with an additional adjustment, so-called \textit{prefreezing} \cite{TS38212}, shows stable performance over a wide range of code rates.

In this paper, we study the fundamentals of puncturing and shortening for binary polar codes.
We show that \textit{binary domination}, formally introduced by Sarkis \textit{et al.} in \cite{Sarkis2016}, plays an important role in determining puncturing and shortening bit patterns in polar coding.
The binary domination relation gives a partial order between two integers, based on their binary representation.
Some properties of binary polar codes constructed from the $2\times 2$ polarization kernel $\left[\begin{smallmatrix} 1 & 0 \\ 1 & 1 \end{smallmatrix}\right]$ can be easily analyzed by the binary representation of the synthetic channel indices.
In fact, the binary domination relation was found to determine a partial order on the reliabilities of the polarized split channels \cite{Mori2009,Schuerch2016}.
Also, it provides guidance on the order of the bits to be punctured or to be shortened.
We first prove that the encoder input bits to be made incapable by puncturing are determined with the partial order by binary domination.
Then, we show that there are puncturing bit patterns identical to or `reverse' to a given incapable bit pattern.
Also, in the shortened polar codes introduced in \cite{Wang2014}, we show that the shortening bit patterns in the encoder input obey the partial order by binary domination.

It has been believed in the literature \cite{Eslami2011,Niu2013a,Shin2013,Zhang2014,Kim2015,Saber2015,Chandesris2017,Hong2017,Hong2018,El-Khamy2018,Wang2014,Miloslavskaya2015,Bioglio2017} that shortening of polar codes inevitably starts from high-indexed bits,
whereas puncturing begins from low-indexed bits.
Even in the CB-RM in 3GPP NR \cite{TS38212}, puncturing is performed from the head of the buffer, whereas shortening is done from the tail of the buffer.
This means that the buffer for puncturing needs to be differently managed from that for shortening.
Unlike such a conventional approach, we show that puncturing does not need to begin from an encoder output bit with low index, but it may start from an encoder output bit with high index.
In addition, we show that puncturing and shortening bit patterns can be aligned by binary domination at the encoder output, so puncturing can be performed in exactly the same order as shortening.
Finally, we propose a practical CB-RM scheme that exploits a unified bit pattern and a unified buffer to support all the rate-matching methods including puncturing, shortening, and repetition.
This scheme can be simply and efficiently implemented for a practical rate-compatible polar coding chain.

The contributions of this work are summarized as follows:
\begin{itemize}
    \item \textit{Fundamentals of punctured polar codes:}
        We identify a necessary and sufficient condition for an encoder input bit to be made incapable by puncturing in Theorem \ref{thm:incapable02}.
        We show that binary domination completely determines an incapable bit pattern.
        We then find all puncturing bit patterns that result in the same incapable bit pattern in Theorem~\ref{thm:eq_punct}.
        Theorems~\ref{thm:id_punct} and \ref{thm:rev_punct} identify two important patterns -- identical and reverse patterns, respectively, among these puncturing bit patterns.
        In particular, the reverse puncturing bit pattern is obtained by bitwise complement of the desired incapable bit pattern.
        This enables us to begin puncturing from a high-indexed encoder output bit, in constrast to the conventional rate-matching schemes for polar codes.
    \item \textit{Fundamentals of shortened polar codes:} We prove in Theorem~\ref{thm:fixed} that any fixed bit pattern at the encoder output, resulting from shortening, is also constrained by binary domination.
        In Corollary~\ref{cor:short}, we verify that the corresponding shortening bit pattern in the encoder input is identical to the fixed bit pattern.
        We also identify a necessary and sufficient condition for an encoder output bit to be fixed by shortening in Theorem~\ref{thm:short}.
    \item \textit{Unified rate matching for polar codes:} We propose a unified circular-buffer rate matching scheme to align the puncturing bit pattern and the fixed bit pattern by shortening via binary domination.
        The proposed scheme reduces the implementation complexity of the rate-matching scheme in 3GPP NR and makes the coded modulation chain much simpler, while keeping the same good performance.
\end{itemize}

The rest of this paper is organized as follows.
Section~\ref{sec:preliminary} describes the preliminaries to polar coding and rate matching.
Sections~\ref{sec:punct} and \ref{sec:short}, respectively, reveal that both puncturing and shortening bit patterns for a polar code are related to binary domination.
Section~\ref{sec:unified} designs a unified rate-matching scheme to support puncturing, shortening, and repetition.
Finally, Section~\ref{sec:conclude} describes future works, and concludes the paper.

\section{Preliminaries} \label{sec:preliminary}

\subsection{Notation}
Throughout the paper, we write calligraphic letters (\textit{e.g.} $\mathcal{A}$) to denote sets.
Conventionally, we use $\mathcal{A}^\mathsf{c}$ to denote the complementary set of $\mathcal{A}$.
Given $\mathcal{A}$ and $\mathcal{B}$, we write $\mathcal{A}\backslash \mathcal{B}$ to denote the relative complement of $\mathcal{A}$ with respect to $\mathcal{B}$.
Given an integer set $\mathcal{A}$ and an integer $b$, we write $b+\mathcal{A}$ to denote $\left\{ b+a \mid a\in\mathcal{A} \right\}$.
For example, $2 + \{0,1\} = \{2,3\}$.
Let $\mathbb{N}$, $\mathbb{Z}$, and $\mathbb{R}$ be the set of natural numbers, integers, and real numbers, respectively.
For a positive integer $m$, we write $\mathbb{Z}_{m}$ to denote the set of integers from 0 to $m-1$, that is, $\mathbb{Z}_m = \{0,1,\ldots,m-1\}$.
Given two integers $i<j$, we use $[i:j]$ to denote the set of consecutive integers from $i$ to $j$, \textit{i.e.}, $[i:j]=\{i,i+1,\ldots,j\}$.

We use boldface lowercase letters (\textit{e.g.} $\mathbf{a}$) and boldface uppercase letters (\textit{e.g.} $\mathbf{A}$) to denote vectors and matrices, respectively.
Given $\mathbf{a} =(a_0,\ldots,a_{N-1})$ and $\mathcal{A} \subseteq \mathbb{Z}_N$,
we write $\mathbf{a}_{\mathcal{A}}$ to denote the subvector $(a_i:i\in\mathcal{A})$, where the order of the entries in $\mathbf{a}_{\mathcal{A}}$ is the same as that in $\mathbf{a}$.

We write italic boldface uppercase letters (\textit{e.g.} $\textbfit{A}$) to denote sequences.
Given a sequence $\textbfit{A}=(a_0,\ldots,a_{N-1})$, we write $\mathcal{A}\leftarrow\textbfit{A}$ to represent that the set $\mathcal{A}$ is constituted by taking the elements of $\textbfit{A}$, that is, $\mathcal{A}=\{a_0,\ldots,a_{N-1}\}$.
Here, the order of elements no longer matters in $\mathcal{A}$.

\begin{figure*}[t]
	\centering
    \if\doccolumn1
    \vspace{-0.3cm}
    \fi
	\includegraphics[width=\textwidth]{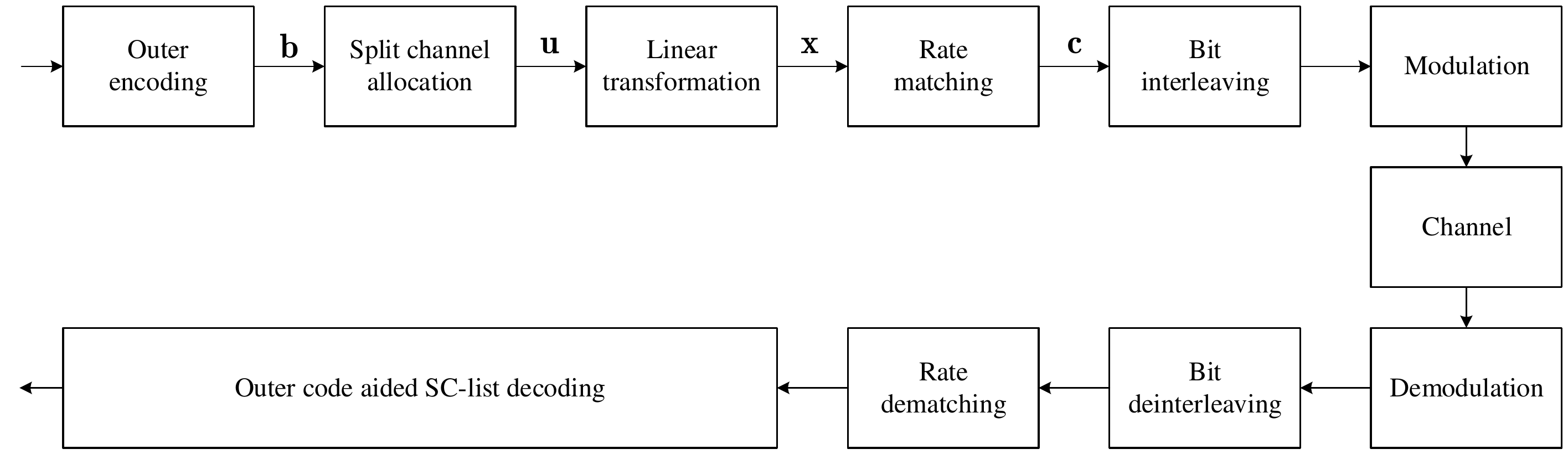}
	\caption{Block diagram of polar coded bit-interleaved coded modulation (BICM) system}
	\label{fig:bicm}
    \if\doccolumn1
    \vspace{-0.5cm}
    \fi
\end{figure*}

\subsection{Encoding of Polar Codes} \label{sec:encoding}

Fig.~\ref{fig:bicm} shows a general bit-interleaved coded modulation (BICM) system with an $(M, K)$ polar code,
where $M$ and $K$ denote the length and the dimension (or the information length), respectively.
The polar code is obtained via puncturing, shortening, and repetition from a mother code of length $N$ and has rate $R=K/M$.

Let $\mathbf{b}\in\mathbb{F}_2^K$ denote the vector of information bits, where $\mathbb{F}_2$ is the binary field.
Generally, $\mathbf{b}$ is a codeword of an outer code such as cyclic redundancy check (CRC) codes \cite{Tal2015,Niu2012}, extended BCH codes \cite{Trifonov2013}, and parity-check codes \cite{Wang2016}.
The outer code increases the minimum distance of the resultant concatenated code, thereby improving the decoding performance when near maximum-likelihood (ML) decoding such as SCL decoding \cite{Tal2015} and SC-stack (SCS) decoding \cite{Niu2012a} is employed.

We assume that $M$ and $K$ are given first, and the other code parameters including $N$ are determined according to $M$ and $K$.
We consider a polar code constructed by the $2 \times 2$ binary polarization kernel $\left[\begin{smallmatrix} 1 & 0 \\ 1 & 1 \end{smallmatrix}\right]$ as a mother code,
so $N=2^n$ for some integer $n \geq 1$.
Usually, $N$ is chosen to be the smallest power of two greater than or equal to $M$, that is, $N=2^{\lceil \log_2 M \rceil}$.
However, if $M$ is rather close to half of $2^{\lceil  \log_2 M \rceil}$,
a polar code of length $2^{\lceil  \log_2 M \rceil - 1}$ would be a good option in terms of performance and complexity.
For example, the NR polar coding scheme chooses
$N = 2^{\lceil  \log_2 M \rceil - 1}$ if $M \leq \frac{9}{8} \times 2^{\lceil \log_2 M \rceil - 1}$ and $R < \frac{9}{16}$; and $N = 2^{\lceil  \log_2 M \rceil}$, otherwise \cite{TS38212}.

Given $N$, the information vector $\mathbf{b}$ is mapped to a subvector of an encoder input vector $\mathbf{u}\in\mathbb{F}_2^{N}$.
This bit mapping is carefully done by considering the quality of each polarized synthetic channel and is called the \textit{split channel allocation}.
The vector $\mathbf{u}$ is divided into three disjoint subvectors $\mathbf{u}_{\mathcal{I}}$, $\mathbf{u}_{\mathcal{F}}$, $\mathbf{u}_{\mathcal{Z}}$,
where $\mathcal{I},\mathcal{F},\mathcal{Z}\subset \mathbb{Z}_N$ are the index sets of information bits, frozen bits, and zero-capacity bits, respectively.
First, $\mathcal{Z}$ is determined by which bits are to be punctured or to be shortened in the rate matcher.
This procedure will be introduced in detail in the next subsection.
Among $\mathbb{Z}_{N}\backslash \mathcal{Z}$,
the indices corresponding to the $K$ most reliable split channels constitute $\mathcal{I}$,
while the remaining indices comprise $\mathcal{F}$.
In the split channel allocation, $\mathbf{b}$ is mapped to $\mathbf{u}_\mathcal{I}$, and both $\mathbf{u}_\mathcal{F}$ and $\mathbf{u}_\mathcal{Z}$ are generally set to the zero vector.

As an example of the split channel allocation, a single sequence of indices, $\textbfit{Q}_{1024}=(q_0,\ldots,q_{1023})$, is used for any combination of $M$ and $K$ in the NR polar coding scheme \cite{TS38212}.
We call this sequence the \textit{NR polar code sequence}.
Given $N\leq1024$, a sub-sequence $\textbfit{Q}_N=(q_i:q_i \in\textbfit{Q}_{1024},q_i<N)$ is extracted from $\textbfit{Q}_{1024}$ while keeping the relative order of elements.
Next, the index set of $J$ zero-capacity bits, $\mathcal{Z}$, is determined by $J$ punctured or shortened bits in the rate matcher.
The complementary index set is then given by $\textbfit{Q}_{N-J}=(q_i:q_i\in\textbfit{Q}_N,q_i \notin \mathcal{Z})$, where the relative order of elements still remains unchanged.
Finally, $\mathcal{I}\leftarrow(\textbfit{Q}_{N-J})_0^{K-1}$ and $\mathcal{F}\leftarrow(\textbfit{Q}_{N-J})_K^{N-J-1}$.

After the split channel allocation, an encoder output vector $\mathbf{x}$ is obtained by the linear transformation, $\mathbf{x} = \mathbf{u} \mathbf{G}_{N}$,
where $\mathbf{G}_{N} \in \mathbb{F}_2^{N \times N}$ is the generator matrix of a polar code of length $N$.
In Ar{\i}kan's original polar coding scheme \cite{Arikan2009}, the linear transformation $\mathbf{B}_{N} \mathbf{F}_2^{\otimes n}$ is considered as $\mathbf{G}_N$,
where $\mathbf{F}_2^{\otimes n}$ denotes the $n$-th Kronecker power of $\mathbf{F}_2=\left[\begin{smallmatrix} 1 & 0 \\ 1 & 1 \end{smallmatrix}\right]$,
and $\mathbf{B}_{N}\in \mathbb{F}_2^{N \times N}$ is the $N$-dimensional bit-reversal permutation matrix.
Since $\mathbf{B}_{N}$ just reorders either $\mathbf{u}$ or $\mathbf{x}$,
we consider
\begin{equation}\label{eq:polar_enc}
    \mathbf{x}=\mathbf{u}\mathbf{F}_2^{\otimes n},
\end{equation}
for simple description throughout the paper, as in the NR polar coding system \cite{TS38212}.

A codeword $\mathbf{c}\in \mathbb{F}_2^M$ is finally obtained from $\mathbf{x}$ via rate matching.
In the rate matcher, $N-M$ bits are punctured from $\mathbf{x}$ if $M < N$,
while $M-N$ bits are additionally selected from $\mathbf{x}$ and appended to $\mathbf{x}$ to generate $\mathbf{c}$ if $M > N$.
In wireless communication systems, a CB-RM scheme with bit interleaving is generally adopted for simple implementation
so that encoder output bits are always circularly extracted in order from a buffer.
In the CB-RM, the bit rearranging order is unified for all the rate-matching operations such as puncturing, shortening, and repetition.
For example, NR exploits a CB-RM scheme with subblock-wise permutation \cite{TS38212}.
It achieves stable performance over all the lengths and rates of interest.
The bits in $\mathbf{c}$ are generally interleaved to improve the performance of BICM,
and finally, they are modulated and transmitted over a physical channel.

An important feature of polar coding is that the split channel allocation and the rate matching are closely related to each other.
More specifically, the overall channel polarization is affected by which and how many bits are punctured or repeated,
and some split channels even become incapable of delivering any information due to puncturing.


\subsection{Rate Matching} \label{sec:rate_matching}

Polar codes based on the $2\times 2$ polarization kernel $\mathbf{F}_2=\left[\begin{smallmatrix} 1 & 0 \\ 1 & 1 \end{smallmatrix}\right]$ have a power of two as their length.
Thus, rate matching is applied to modify the length and adjust the rate.
There are three major types of rate-matching techniques for polar codes: puncturing, shortening, and repetition.
These rate-matching schemes alter overall channel polarization.
In particular, when puncturing or shortening occurs,
some encoder input bits are not available to deliver information since their symmetric capacity becomes zero.
Fig.~\ref{fig:rm} briefly describes a subsequent consequence that appears on the opposite side, due to puncturing or shortening a polar code.

\begin{figure}[t]
	\centering
    \if\doccolumn1
    \vspace{-0.3cm}
    \fi
    \includegraphics[width=8.85cm]{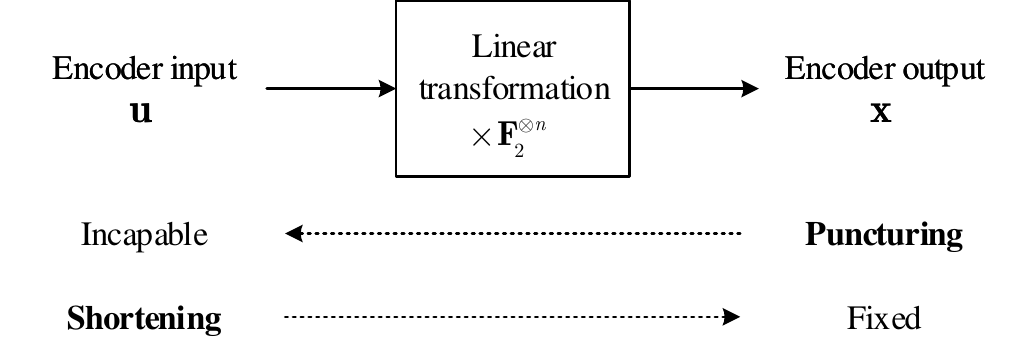}
	\caption{Puncturing and shortening of polar codes. Puncturing some encoder output bits results in corresponding incapable bits at the encoder input. Shortening some encoder input bits, if carefully chosen, makes the corresponding encoder output bits have a deterministic value.
Incapable and shortened bits are not available to transmit information, and they belong to a class of zero-capacity bits.}
	\label{fig:rm}
    \if\doccolumn1
    \vspace{-0.5cm}
    \fi
\end{figure}

\subsubsection{Puncturing}

Puncturing removes $J$ bits from $\mathbf{x}$ to make a codeword of length $M=N-J$,
and these punctured bits are not transmitted.
The decoder on the receiver side does not have any stochastic information about them,
so their initial log-likelihood ratios (LLRs) are set to zero.
The impact of puncturing on channel polarization is more significant, compared with other rate matching schemes.
In particular, when coded bits are punctured, the symmetric capacities of some split channels are degraded to zero.
The encoder input bits in $\mathbf{u}$ that correspond to the zero-capacity synthetic channels due to puncturing are called the \textit{incapable bits} \cite{Shin2013}.
The values of these incapable bits are typically set to zero.
Let $\mathcal{X}_p \subset \mathbb{Z}_{N}$ be the index set of punctured bits in $\mathbf{x}$,
and let $\mathcal{U}_p \subset \mathbb{Z}_{N}$ denote the index set of the corresponding incapable bits in $\mathbf{u}$.
It was shown in \cite{Shin2013} that any puncturing bit pattern with $J$ punctured bits in $\mathbf{x}$ leads to exactly the same number of incapable bits in $\mathbf{u}$,
and hence, $|\mathcal{U}_p|=|\mathcal{X}_p|$.

\subsubsection{Shortening}
In \cite{Wang2014}, a shortening technique for polar codes was introduced.
Due to the lower-triangular form of $\mathbf{F}_2^{\otimes n}$,
if the encoder carefully chooses $J$ bits in $\mathbf{u}$ and fixes their values to zero,
then there are exactly $J$ bits in $\mathbf{x}$ whose values are zero.
These $J$ bits in $\mathbf{x}$ are not transmitted,
but the decoder on the receiver side is aware that their values are zero.
Thus, the LLR values corresponding to these fixed bits with shortening are initially set to infinity before decoding.
The change in channel polarization due to shortening is totally different from that due to puncturing.
It was empirically observed that when the code rate is high, shortening leads to better performance than puncturing.
The bits in $\mathbf{u}$ whose values are set to zero are called the \textit{shortened bits},
and the corresponding bits in $\mathbf{x}$ whose values are fixed to zero by shortened bits are called the \textit{fixed bits}.
The shortened bits also belong to a class of zero-capacity bits because they do not convey any information.
Let $\mathcal{U}_s \subset \mathbb{Z}_{N}$ be the index set of shortened bits in $\mathbf{u}$,
and let $\mathcal{X}_s \subset \mathbb{Z}_{N}$ be the index set of the corresponding fixed bits with deterministic values in $\mathbf{x}$.
It was shown in \cite{Wang2014} that a greedy selection method with the weight-one column criterion makes $\mathcal{U}_s$ and $\mathcal{X}_s$ have the same cardinality, that is, $|\mathcal{X}_s|=|\mathcal{U}_s|$.
While $\mathcal{X}_p$ determines $\mathcal{U}_p$ in puncturing,
$\mathcal{U}_s$ identifies $\mathcal{X}_s$ in shortening.

\subsubsection{Repetition}

If $M$ is slightly larger than $2^{\lceil \log_2 M \rceil -1}$,
then it would be a good option to exploit a polar code of length $N=2^{\lceil \log_2 M \rceil -1}<M$ and repetition.
Compared with the case of choosing a polar code of length $2^{\lceil \log_2 M \rceil}$,
the performance degradation due to excessive puncturing can be avoided;
and as well, the encoding and decoding complexity can be reduced to about half.
At the transmitter, $M-N$ bits in $\mathbf{c}$ are generated by repeating some bits in $\mathbf{x}$;
and before decoding at the receiver, the LLRs of the repeated bits are combined with those of the corresponding original bits.
LLR combining improves the effective bitwise channels corresponding to the repeated bits,
so the choice of a repetition pattern affects the decoding performance.

\subsection{Binary Domination}

For an integer $i\in\mathbb{Z}_{2^n}$, let $\langle i \rangle _2 \triangleq (i_{n-1}i_{n-2}\ldots i_0)$ denote the binary representation of $i = \sum_{t=0}^{n-1}i_t 2^t$ where $i_t\in\{0,1\}$ for $t\in[0:n-1]$.
Let $d_\mathrm{H}(i)$ be the Hamming weight of $\langle i \rangle _2$, that is, $d_\mathrm{H}(i)=\sum_{t=0}^{n-1} i_t$.
We write $\bar{i}$ to denote the bitwise complement of $i$, which is obtained by inverting $i_t$ for all $t\in[0:n-1]$.
Clearly, $\bar{i}=(2^n-1)-i$.

\begin{define}[\textit{Binary domination \cite{Sarkis2016}}]\label{def:binary_domination}
    For any two integers $i,j\in\mathbb{Z}_{2^n}$, we say that $j$ dominates $i$ or $i$ is dominated by $j$, denoted by $i \preceq j$ (or $j \succeq i$), if $i_t \leq j_t$ for all $t\in[0:n-1]$.
    Furthermore, we say that $j$ strictly dominates $i$ or $i$ is strictly dominated by $j$, denoted by $i \prec j$ (or $j \succ i$), if $i \preceq j$ but $i \neq j$.
    The relations $\preceq$ and $\prec$ are called the binary domination relation and the strict binary domination relation, respectively.
\end{define}

The binary domination relations give a partial order on $\mathbb{Z}_{2^n}$.
For example, in $\mathbb{Z}_{2^4}$, we have $ 5 \prec 13$
since $\langle 5 \rangle_2 = (0101)$ and $\langle 13 \rangle_2 = (1101)$.
However, $7$ and $8$ are not comparable because $\langle 7 \rangle_2 = (0111)$ and $\langle 8 \rangle_2 = (1000)$.
Note that the binary domination relation $\preceq$ is reflexive, antisymmetric, and transitive,
while the strict binary domination relation $\prec$ is irreflexive, asymmetric, and transitive \cite{Djeraba2008}.
In addition, it is clear that $\bar{j} \prec \bar{i}$ if $i\prec j$.

Mori and Tanaka observed a channel-independent phenomenon that $i\prec j$ guarantees that the reliability of $u_j$ is usually better and never worse than that of $u_i$ in channel polarization \cite{Mori2009}.
That is, the binary domination relation defines a partial order on the reliabilities of the polarized split channels, regardless of the channel statistics.
This ordering is found without any complicated analysis such as density evolution \cite{Mori2009,Schuerch2016}.
Well-designed sequences defining polar codes should adhere to the partial order by binary domination.
The binary domination relation was also used in the efficient implementation of systematic polar codes \cite{Sarkis2016}.


Based on binary domination, we define a dominated integer set and a dominating integer set as follows.

\begin{define} [\textit{Dominated integer set}] \label{def:dominated_set}
The dominated integer set of $i \in \mathbb{Z}_{2^n}$ is defined as
$\mathcal{D}_i \triangleq \{k\in\mathbb{Z}_{2^n}\mid k \preceq i \}$.
The strictly dominated integer set of $i \in \mathbb{Z}_{2^n}$ is defined as
$\hat{\mathcal{D}}_i \triangleq \{k\in\mathbb{Z}_{2^n}\mid k \prec i \}$.
\end{define}

\begin{define} [\textit{Dominating integer set}] \label{def:dominating_set}
The dominating integer set of $i\in\mathbb{Z}_{2^n}$ is defined as
$\mathcal{G}_i \triangleq \{k\in\mathbb{Z}_{2^n}\mid k \succeq i \}$.
The strictly dominating integer set of $i\in\mathbb{Z}_{2^n}$ is defined as
$\hat{\mathcal{G}}_i \triangleq \{k\in\mathbb{Z}_{2^n}\mid k \succ i \}$.
\end{define}

Clearly,
$\mathcal{D}_i = \hat{\mathcal{D}}_i \cup \{i\}$,
$\hat{\mathcal{D}}_i = \mathcal{D}_i \backslash \{i\}$,
$\mathcal{G}_i = \hat{\mathcal{G}}_i \cup \{i\}$, and
$\hat{\mathcal{G}}_i = \mathcal{G} \backslash \{i\}$ by definitions.
The partial order by binary domination plays a decisive role in determining puncturing and shortening bit patterns for polar codes, as will be shown in the next sections.
For a clear presentation,
we define a partially-ordered sequence on $\mathbb{Z}_{2^n}$, called a $2^n$-\textit{posequence} for short, as a sequence whose entries follow the partial order by binary domination.

\begin{define}[\textit{$2^n$-posequence}]\label{def:posequence}
A sequence $(p_0,\ldots,p_{2^n-1})$ is called a $2^n$-posequence if it is a permutation on $\mathbb{Z}_{2^n}$ such that there is no pair of $i$ and $j$ in $\mathbb{Z}_{2^n}$ with $p_i \succ p_j$,
that is, either $p_i\prec p_j$, or $p_i$ and $p_j$ are not comparable for any $i<j$.
\end{define}

For example, for $n=2$, $(0,2,1,3)$ is a $2^2$-posequence whereas $(0,1,3,2)$ is not since $3 \nprec 2$.
Within any $2^n$-posequence, all the integers dominated by $j$ appear ahead of $j$,
whereas all the integers that dominate $j$ are behind $j$.

\begin{define}\label{def:comply_with_bd}
A set $\mathcal{A}\subset \mathbb{Z}_{2^n}$ is said to comply with binary domination if
either $\mathcal{D}_j\subset\mathcal{A}$ for all $j\in\mathcal{A}$
or $\mathcal{G}_j\subset\mathcal{A}$ for all $j\in\mathcal{A}$.
\end{define}

\begin{define}\label{def:rm_pattern}
A set $\mathcal{A}\subset\mathbb{Z}_{2^n}$ is called a puncturing (incapable, shortening, and fixed, respectively) bit pattern if the bits with indices in $\mathcal{A}$ are punctured (incapable, shortened, and fixed, respectively), while the other bits remain unchanged.
\end{define}

\section{Puncturing and Incapable Bit Patterns}\label{sec:punct}

The relation between puncturing bit patterns and their corresponding incapable ones may be well understood by exploring the operation of SC decoding.
To reveal this, this section first reviews SC decoding.
Readers may refer to \cite{Arikan2009,Sasoglu2012,Niu2014,Tal2015} for better comprehension of SC and SCL decoding.
Based on the SC decoding operation, the possible incapable bit patterns and their corresponding puncturing ones are then investigated.

\subsection{Successive-Cancellation Decoding}

SC decoding for a polar code can be regarded as belief-propagation (BP) over the bipartite graph corresponding to the generator matrix $\mathbf{G}_N=\mathbf{F}_2^{\otimes n}$ in \eqref{eq:polar_enc}.
Fig.~\ref{fig:dec} depicts an example of the decoding of a polar code of length $N=8$ over the corresponding bipartite graph.
The graph consists of variable nodes (circles), check nodes (squares), and edges.
A variable node corresponds to a single bit, while a check node represents a linear constraint that the binary sum of the values of all neighbor variable nodes is equal to zero.

Let $\mathds{G}_{2^n}$ denote the graph corresponding to the polar code of length $N=2^n$.
The graph is divided into $n + 1$ stages indexed from $0$ (leftmost) to $n$ (rightmost).
Let $v_i^{(t)}$ and $c_j^{(t)}$ denote the $i$-th variable node and the $j$-th check node at stage $t$, respectively.
Then, the graph $\mathds{G}_{2^n}$ consists of
the sets of variable nodes $\mathcal{V}_{\mathbb{Z}_{2^n}}^{(t)}$ for $t=[0:n]$,
the sets of check nodes $\mathcal{C}_{\mathbb{Z}_{2^n}}^{(l)}$ for $l=[0:n-1]$,
and the edges connecting these nodes,
where $\mathcal{V}_{\mathbb{Z}_{2^n}}^{(t)}=\left\{v_i^{(t)}\mid i\in\mathbb{Z}_{2^n}\right\}$
and $\mathcal{C}_{\mathbb{Z}_{2^n}}^{(l)}=\left\{c_j^{(l)}\mid j\in\mathbb{Z}_{2^n}\right\}$
denote the set of $2^n$ variable nodes at stage $t$ and the set of $2^n$ check nodes at stage $l$, respectively.
In particular, $\mathcal{V}_{\mathbb{Z}_{2^n}}^{(0)}$ and $\mathcal{V}_{\mathbb{Z}_{2^n}}^{(n)}$ correspond to an encoder input vector $\mathbf{u}$  and an encoder output vector $\mathbf{x}$, respectively.

For $t\in[0:n-1]$, variable nodes in $\mathcal{V}_{\mathbb{Z}_{2^{n}}}^{(t)}$ and $\mathcal{V}_{\mathbb{Z}_{2^{n}}}^{(t+1)}$ are connected to check nodes in $\mathcal{C}_{\mathbb{Z}_{2^{n}}}^{(t)}$,
and the connections are determined by the nonzero entries of $\mathbf{F}_2^{\otimes n}$.
Due to the recursive construction of $\mathbf{F}_2^{\otimes n}$, the connection of the graph can be characterized by a simple rule with respect to the binary representation of the corresponding check node index.
Recall that $i_t$ denotes the $t$-th component of the binary representation $\langle i \rangle_2 = (i_{n-1}i_{n-2}\cdots i_0)$ of $i\in\mathbb{Z}_{2^n}$.
Note that $c_i^{(t)}$ with $i_t=0$ is connected to three variable nodes $v_i^{(t)}$, $v_{i+2^t}^{(t)}$, and $v_i^{(t+1)}$,
while $c_i^{(t)}$ with $i_t=1$ is connected to two variable nodes $v_i^{(t)}$ and $v_i^{(t+1)}$.

\begin{figure}[t]
	\centering
    \if\doccolumn1
    \vspace{-0.3cm}
    \fi
    \includegraphics[width=8.85cm]{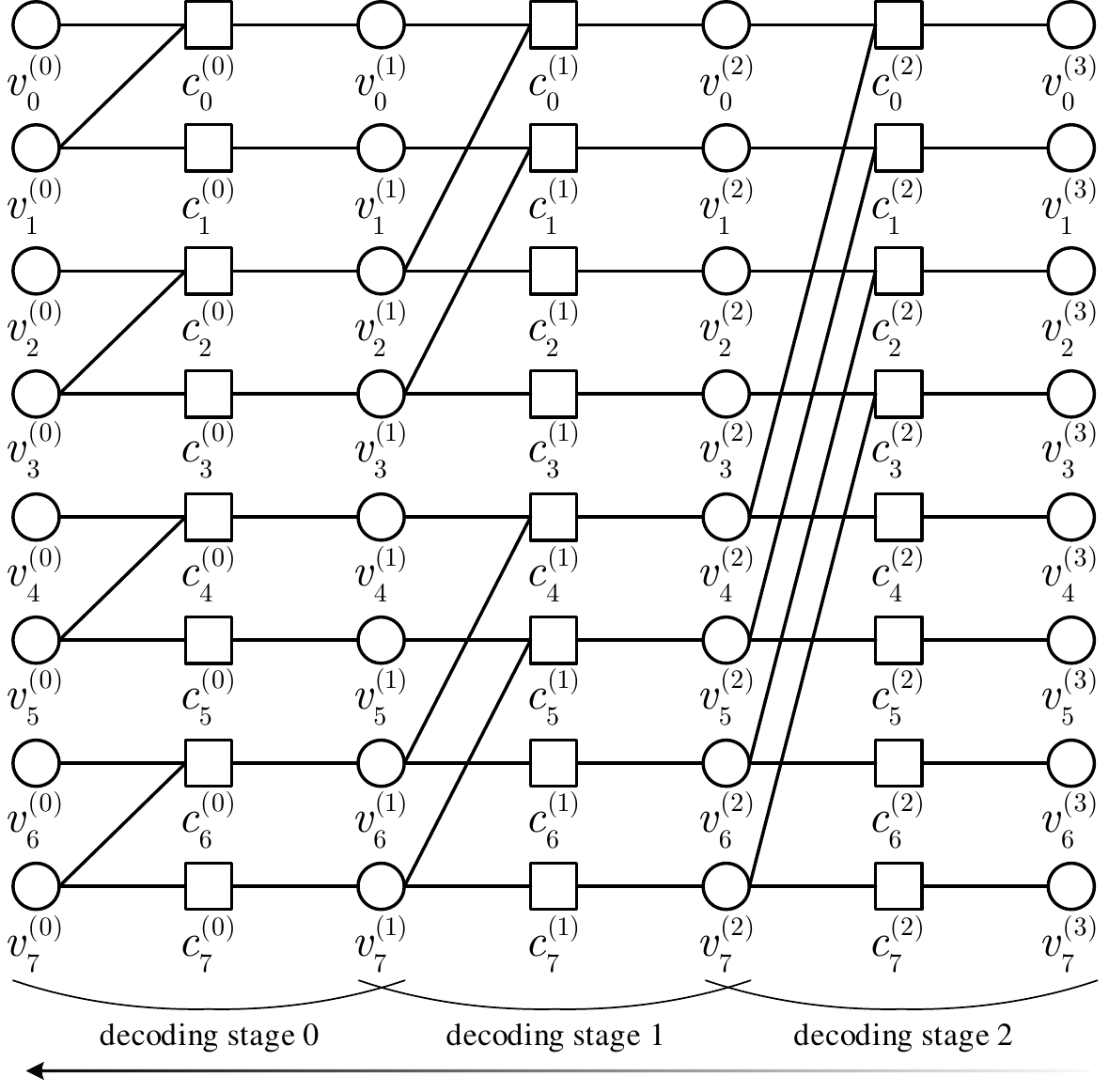}
    \if\doccolumn1
    \vspace{-0.5cm}
    \fi
	\caption{Decoding graph $\mathds{G}_8$ based on $\mathbf{F}_2^{\otimes 3}$ of a polar code of $N=8$.
    }
	\label{fig:dec}
    \if\doccolumn1
    \vspace{-0.5cm}
    \fi
\end{figure}

In SC decoding over the graph, an LLR $\alpha_i^{(t)}\in\mathbb{R}$ and a hard-decision value $\beta_i^{(t)}\in\mathbb{F}_2$ are calculated for each $v_i^{(t)}$.
The LLR values are calculated from right to left in the graph,
while the hard-decision values based on the currently estimated encoding input bits are passed from left to right.
We focus only on the update of LLR values.
Let $\ell_i$ denote the intrinsic LLR for $x_i$, which is calculated from the corresponding received symbol.
While $\ell_i=0$ if $x_i$ is punctured, $\ell_i=\infty$ if $x_i$ is fixed by shortening.
The intrinsic LLRs are fed into the variable nodes at stage $n$ such that $\alpha_i^{(n)} = \ell_i$ for all $i\in \mathbb{Z}_{2^n}$,
and then, the LLRs $\alpha_i^{(t)}$ are calculated from stage $t=n-1$ to stage $t=0$.
In each LLR calculation, one of two functions is adaptively used depending on the variable node index.
At stage $t$, $\alpha_i^{(t)}$ with $i_{t}=0$ is updated by the function $f:\mathbb{R}\times \mathbb{R}\rightarrow \mathbb{R}$ as
\begin{equation}\label{eq:f_func}
\begin{split}
    \alpha_i^{(t)}
    &= f\left( \alpha_i^{(t+1)}, \alpha_{i+2^{t}}^{(t+1)} \right) \\
    &\triangleq 2 \tanh^{-1} \left( \tanh\left(\tfrac{1}{2} \alpha_i^{(t+1)}\right) \tanh\left(\tfrac{1}{2} \alpha_{i+2^{t}}^{(t+1)}\right) \right).
\end{split}
\end{equation}
Note that $\alpha_i^{(t)}=0$ in \eqref{eq:f_func} if any one of $\alpha_i^{(t+1)}$ and $\alpha_{i+2^t}^{(t+1)}$ is equal to zero.
On the other hand, $\alpha_i^{(t)}$ with $i_{t}=1$ is updated by the function $g:\mathbb{R}\times \mathbb{R} \times \{0,1\} \rightarrow \mathbb{R}$ as
\begin{equation}\label{eq:g_func}
\begin{split}
    \alpha_i^{(t)} &= g\left( \alpha_{i-2^t}^{(t+1)}, \alpha_{i}^{(t+1)}, \beta_{i-2^{t}}^{(t)} \right) \\
    &\triangleq \left(1-2\beta_{i-2^{t}}^{(t)}\right)\alpha_{i-2^{t}}^{(t+1)} + \alpha_i^{(t+1)},
\end{split}
\end{equation}
where $\beta_{i-2^t}^{(t)}$ is the hard-decision value of $v_{i-2^t}^{(t)}$, previously estimated by the successive cancellation procedure.
Clearly, $\alpha_i^{(t)}=0$ in \eqref{eq:g_func} only if both $\alpha_i^{(t+1)}$ and $\alpha_{i-2^t}^{(t+1)}$ are zero.

\subsection{Incapable Bit Patterns}

An encoder input bit is said to be incapable if its corresponding LLR becomes zero when SC decoding is employed for a punctured polar code.
Throughout this section, consider puncturing a polar code of length $N=2^n$.
For $\mathcal{A}\subset\mathbb{Z}_{2^n}$ and $t\in[0:n]$,
we write $\mathcal{P}_{\mathcal{A}}^{(t)}$ to denote the index set of variable nodes with LLR zero at stage $t$ restricted to $\mathcal{A}$, that is, $\mathcal{P}_{\mathcal{A}}^{(t)} \triangleq \left\{ k\in\mathcal{A} \mid \alpha_{k}^{(t)}=0 \right\}$.
In particular, $\mathcal{P}_{\mathbb{Z}_{2^n}}^{(0)}=\mathcal{U}_p$ and $\mathcal{P}_{\mathbb{Z}_{2^n}}^{(n)}=\mathcal{X}_p$, where $\mathcal{U}_p$ and $\mathcal{X}_p$ denote the index set of incapable bits and the index set of punctured bits, respectively.
An interesting fact shown in \cite{Shin2013} is that the number of incapable split channels is exactly the same as the number of punctured bits, \textit{i.e.}, $|\mathcal{P}_{\mathbb{Z}_{2^n}}^{(0)}|=|\mathcal{P}_{\mathbb{Z}_{2^n}}^{(n)}|$, regardless of the puncturing bit pattern.
In the following theorem, we show that $\mathcal{U}_p$ complies with binary domination.

\begin{thm}\label{thm:incapable01}
    For any $j\in\mathbb{Z}_{2^n}$, if $j\in\mathcal{U}_p$, then $\mathcal{D}_j \subset \mathcal{U}_p$.
    In other words, $\mathcal{U}_p$ complies with binary domination.
\end{thm}
\begin{proof}
    Let $I(W_{N}^{(i)})$ denote the symmetric capacity of the $i$-th synthetic channel in the polar code of length $N=2^n$
    when SC decoding is applied.
    It was shown in \cite{Mori2009} that $I(W_{N}^{(i)}) \leq I(W_{N}^{(j)})$ if $i \prec j$, regardless of the channel statistics.
    If $I(W_{N}^{(j)})=0$ for some reason, then $I(W_{N}^{(i)})=0$ for all $i\prec j$.
    The symmetric capacity of an incapable bit due to puncturing is zero.
    Thus, if $u_j$ is made incapable, then $u_i$ is also made incapable for all $i \prec j$.
\end{proof}

The following three lemmas present some additional properties of puncturing a polar code.

\begin{lem}\label{lem:incapable01}
    For any $t\in[0:n]$ and any $\ell\in[0:2^{n-t}-1]$, we have
    \begin{equation}\label{eq:incap_lem1}\nonumber
    \left|\mathcal{P}_{\ell\cdot2^t+\mathbb{Z}_{2^t}}^{(0)}\right|
    =\left|\mathcal{P}_{\ell\cdot2^t+\mathbb{Z}_{2^t}}^{(1)}\right|
    =\ldots
    =\left|\mathcal{P}_{\ell\cdot2^t+\mathbb{Z}_{2^t}}^{(t)}\right|.
    \end{equation}
\end{lem}

\begin{proof}
    For $\mathcal{A}\subset\mathbb{Z}_{2^n}$,
    let $\mathcal{V}_{\mathcal{A}}^{(t)}\subset\mathcal{V}_{\mathbb{Z}_{2^n}}^{(t)}$
    and $\mathcal{C}_{\mathcal{A}}^{(t)}\subset \mathcal{C}_{\mathbb{Z}_{2^n}}^{(t)}$
    denote the set of variable nodes and the set of check nodes at stage $t$, respectively.
    For $t\in[0:n]$ and $\ell\in[0:2^{n-t}-1]$,
    we write $\mathds{G}_{2^t}(l)$ to denote the subgraph restricted to
    $\bigcup_{i=0}^{t-1} \big{(}
    \mathcal{V}_{l\cdot2^t+\mathbb{Z}_{2^t}}^{(i)} \cup \mathcal{C}_{l\cdot2^t+\mathbb{Z}_{2^t}}^{(i)}
    \big{)}
    \cup
    \mathcal{V}_{l\cdot2^t+\mathbb{Z}_{2^t}}^{(t)}$.
    Due to the recursive construction of a polar code,
    $\mathds{G}_{2^t}(l)$ can be viewed as the graph corresponding to an independent polar code of length $2^t\leq 2^n$.
    Hence, by applying the result in \cite{Shin2013} to this code, we conclude that the number of variable nodes with LLR zero at each stage is the same as the number of punctured bits to this code.
\end{proof}

Note that $|\mathcal{P}_{\mathbb{Z}_{2^n}}^{(0)}|=|\mathcal{P}_{\mathbb{Z}_{2^n}}^{(n)}|$
in Lemma~\ref{lem:incapable01} when $t=n$ and $\ell=0$.
Therefore, Lemma~\ref{lem:incapable01} is a generalized version of the related result in \cite{Shin2013}.

\begin{lem}\label{lem:incapable02}
    For any $t\in[0:n-1]$ and $\ell\in[0:2^{n-t-1}-1]$,
    we have $\mathcal{P}_{(2\ell+1)\cdot2^t+\mathbb{Z}_{2^t}}^{(t)} \subset 2^t+\mathcal{P}_{2\ell\cdot2^t+\mathbb{Z}_{2^t}}^{(t)}$.
\end{lem}
\begin{proof}
    See Appendix~\ref{proof:lem:incapable02}.
\end{proof}

\begin{lem}\label{lem:incapable03}
    Assume that $\big{|}\alpha_i^{(t)}\big{|} > 0$ for $i\in\mathbb{Z}_{2^t}$ at stage $t\in[0:n]$ after SC decoding.
    Then, there exists an index $k\in\mathbb{Z}_{2^n}$ such that additional puncturing of $x_k$ (\textit{i.e.}, $\alpha_k^{(n)} \leftarrow 0$) results in  $\alpha_i^{(t)}=0$,
    regardless of how the polar code is currently punctured.
\end{lem}
\begin{proof}
    See Appendix~\ref{proof:lem:incapable03}.
\end{proof}

Based on the above lemmas, we give a necessary and sufficient condition for an encoder input bit to be made incapable by additionally puncturing a single output bit.

\begin{thm}\label{thm:incapable02}
    Let $j\in\mathbb{Z}_{2^n} \backslash \mathcal{U}_p$.
    The encoder input bit $u_j$ can be made incapable by additionally puncturing a single encoder output bit if and only if $\hat{\mathcal{D}}_j \subset \mathcal{U}_p$.
\end{thm}
\begin{proof}
    If $u_j$ is additionally made incapable, then $\mathcal{D}_j \subset \mathcal{U}_p \cup \{j\}$ by Theorem~\ref{thm:incapable01}.
    Therefore, we have $\hat{\mathcal{D}}_j \subset \mathcal{U}_p$.

    The converse is proved by mathematical induction.
    First, note that $\mathbb{Z}_{2^n}$ can be decomposed into $\mathbb{Z}_{2^n}=\{0\}\cup\bigcup_{t=0}^{n-1}\left(2^t+\mathbb{Z}_{2^t}\right)=\mathbb{Z}_{2^{n-1}}\cup\left(2^{n-1}+\mathbb{Z}_{2^{n-1}}\right)$.
    Clearly, $u_0$ can be made incapable by puncturing a single encoder output bit by Lemma~\ref{lem:incapable03} even if no bits have been made incapable earlier.
    We prove that for $j\in2^t+\mathbb{Z}_{2^t}$ with $t\in[0:n-1]$, $u_j$ can be made incapable by additionally puncturing a single encoder output bit if $\hat{\mathcal{D}}_j \subset \mathcal{U}_p$.

    For $t=0$, it suffices to consider only $j=1$ because $j\in2^0+\mathbb{Z}_{2^0}=\{1\}$.
    By assumption, we have $\hat{\mathcal{D}}_1=\{0\} \subset \mathcal{U}_p$.
    Since $\alpha_0^{(0)}=0$ and $\alpha_1^{(0)}\neq0$,
    one of $\alpha_0^{(1)}$ and $\alpha_1^{(1)}$ is zero and the other is nonzero.
    By Lemma~\ref{lem:incapable03}, the nonzero one, either $\alpha_0^{(1)}$ or $\alpha_1^{(1)}$, can be made zero by additionally puncturing a single encoder output bit since $\{0,1\}=\mathbb{Z}_{2^1}$.

    Assume that for any $j\in2^{T-1}+\mathbb{Z}_{2^{T-1}}$ with $1\leq T \leq n-1$, $u_j$ can be made incapable by additionally puncturing a single encoder output bit if $\hat{\mathcal{D}}_j\subset\mathcal{U}_p$.
    Then, by Lemma~\ref{lem:incapable01},
    there exists $\ell\in\mathbb{Z}_{2^T}$ such that $\alpha_\ell^{(T)}=0$ leads to making $u_j$ incapable.
    Furthermore, it follows from Lemma~\ref{lem:incapable03} that $\alpha_\ell^{(T)}$ can be made zero for any $\ell\in\mathbb{Z}_{2^T}$ by puncturing an additional encoder output bit.

    Now, consider $j\in 2^T+\mathbb{Z}_{2^T}$.
    Then, $\hat{\mathcal{D}}_j\subset\mathcal{U}_p$ by assumption and $\hat{\mathcal{D}}_j=\mathcal{D}_{j-2^T}\cup\left(2^T+\hat{\mathcal{D}}_{j-2^T}\right)$.
    As defined in the Proof of Lemma~\ref{lem:incapable01}, $\mathds{G}_{2^T}(0)$ and $\mathds{G}_{2^T}(1)$ represent identical and independent component polar codes of length $2^T$, respectively.
    In the component polar code represented by $\mathds{G}_{2^T}(0)$,
    the zero-LLR pattern $\mathcal{P}_{\mathbb{Z}_{2^T}}^{(T)}$ is involved in generating the incapable bit pattern $\mathcal{P}_{\mathbb{Z}_{2^T}}^{(0)}\supset\mathcal{D}_{j-2^T}$.
    On the other hand, the zero-LLR pattern $\mathcal{P}_{2^T+\mathbb{Z}_{2^T}}^{(T)}$ contributes to making the incapable bit pattern $\mathcal{P}_{2^T+\mathbb{Z}_{2^T}}^{(0)}\supset 2^T+\hat{\mathcal{D}}_{j-2^T}$ in the component polar code corresponding to $\mathds{G}_{2^T}(1)$.
    At stage $T$, we have $-2^T+\mathcal{P}_{2^T+\mathbb{Z}_{2^T}}^{(T)} \subset \mathcal{P}_{\mathbb{Z}_{2^T}}^{(T)}$ by Lemma~\ref{lem:incapable02}.
    Because the zero-LLR pattern $-2^T+\mathcal{P}_{2^T+\mathbb{Z}_{2^T}}^{(T)}$ does not make $u_j$ incapable in $\mathds{G}_{2^T}(1)$,
    the induction hypothesis guarantees that
    there exists $k\in\mathcal{P}_{\mathbb{Z}_{2^T}}^{(T)}\backslash\left(-2^T+\mathcal{P}_{2^T+\mathbb{Z}_{2^T}}^{(T)}\right)$ such that $\alpha_k^{(T)}=0$ leads to making $u_{j-2^T}$ incapable in $\mathds{G}_{2^T}(0)$. 
    In $\mathds{G}_{2^T}(1)$,
    we have $2^T+\hat{\mathcal{D}}_{j-2^T}\subset\mathcal{P}_{2^T+\mathbb{Z}_{2^T}}^{(0)}$ caused by the identical zero-LLR pattern $\mathcal{P}_{2^T+\mathbb{Z}_{2^T}}^{(T)} \subset 2^T+\mathcal{P}_{\mathbb{Z}_{2^T}}^{(T)}$, as in $\mathds{G}_{2^T}(0)$.
    Due to the same structure of $\mathds{G}_{2^T}(0)$ and $\mathds{G}_{2^T}(1)$,
    $u_j$ can be made incapable by additionally setting $\alpha_{k+2^T}^{(T)}=0$.
    Since $\alpha_k^{(T)}=0$ and $|\alpha_{k+2^T}^{(T)}|>0$, one of $\alpha_k^{(T+1)}$ and $\alpha_{k+2^T}^{(T+1)}$ is zero and the other is nonzero.
    Clearly, $\alpha_{k+2^T}^{(T)}$ can be made zero by setting the nonzero one, either $\alpha_{k}^{(T+1)}$ or $\alpha_{k+2^T}^{(T+1)}$, to zero.
    Since $k,k+2^T\in\mathbb{Z}_{2^{T+1}}$,
    the nonzero one can be made zero by additionally puncturing a single encoder output bit by Lemma~\ref{lem:incapable03}.
    Hence, the statement holds for any $j\in\mathbb{Z}_{2^n}$ by mathematical induction.
\end{proof}

It is shown in Theorem~\ref{thm:incapable01} that for an index $j\in\mathcal{U}_p$,
all the indices dominated by $j$ are ahead of $j$ in $\mathcal{U}_p$, while the indices dominating $j$ appear behind of $j$ in $\mathcal{U}_p$.
This is the property that any $2^n$-posequence $\textbfit{P}=(p_0,\ldots,p_{2^n-1})$ has.
Thus, the set $\mathcal{A} \leftarrow (\textbfit{P})_0^{J-1}$ is an achievable incapable bit pattern of length $J$.
In addition, Theorem~\ref{thm:incapable02} demonstrates that
the order of the encoder input bits to be made incapable is restricted only by binary domination.
Definitely, $u_0$ is the first encoder input bit to be incapable since $0$ is dominated by any nonzero integer in $\mathbb{Z}_{2^n}$ with respect to binary domination.

\subsection{Puncturing Bit Patterns}

In this subsection, we investigate puncturing bit patterns that result in a given incapable bit pattern.
First, we identify the sets of encoder output bits required to be punctured to make a certain single encoder input bit incapable.

For clear presentation, we introduce some set notations.
Given two ordered sets $\mathcal{A}=\{a_0,\ldots,a_{m-1}\}$ and $\mathcal{B}=\{b_0,\ldots,b_{m-1}\}$ with the same cardinality $m$,
we define $\mathcal{A}\oplus\mathcal{B}$ as the elementwise addition of these sets, that is, $\mathcal{A}\oplus\mathcal{B}=\{a_0+b_0,\ldots,a_{m-1}+b_{m-1}\}$.
In addition, given a set $\mathcal{A}$, let $P_{m}(\mathcal{A})$ denote the family of the ordered sets obtained by taking $m$ elements from $\mathcal{A}$ with repetition, \textit{i.e.}, $P_{m}(\mathcal{A})
\triangleq
\left\{
\left\{p_0,\ldots,p_{m-1}\right\} \mid p_i\in\mathcal{A}, i=0,\ldots,m-1
\right\}$.
For example, $P_2(\{0, 4\}) = \left\{ \{0,0\},\{0,4\},\{4,0\},\{4,4\}  \right\}$.

\begin{lem}\label{lem:puncturing01}
    For $t\in[0:n]$ and $j\in\mathbb{Z}_{2^n}$, let $\psi_j^{(t)}$ denote the family of minimal sets of variable node indices with LLR zero at stage $t$, which are required to make $u_j$ incapable.
    Then, $\psi_j^{(0)}=\{\{j\}\}$ and $\psi_j^{(n)}$ is obtained by recursively performing
    \begin{equation}\label{eq:psi_op}
        \psi_{j}^{(t+1)}=
        \begin{cases}
        \left\{ \mathcal{B}\oplus\mathcal{Q} \mid \mathcal{Q}\in\psi_j^{(t)},\mathcal{B}\in P_{|\mathcal{Q}|}(\{0,2^t\}) \right\}, & \text{if } j_t = 0  \\
        \left\{ \mathcal{Q}\cup\left(-2^t+\mathcal{Q}\right) \mid \mathcal{Q}\in\psi_j^{(t)} \right\}, & \text{if } j_t = 1
        \end{cases}
    \end{equation}
    for $t=[0:n-1]$.
\end{lem}
\begin{proof}
    See Appendix~\ref{proof:lem:puncturing01}.
\end{proof}

In particular, $\psi_j\triangleq\psi_j^{(n)}$ is simply called the family of minimal puncturing bit patterns required to make $u_j$ incapable.
The cardinalities of $\psi_j$ and its element sets can be found by Lemma~\ref{lem:puncturing01}.
For $t\in[0:n]$ and $j\in\mathbb{Z}_{2^n}$, let $\mathcal{Q}_j^{(t)}$ denote an element set of $\psi_j^{(t)}$.
Starting with $|\mathcal{Q}_j^{(0)}|=|\{j\}|=1$, we have
\begin{equation}\label{eq:size_Q_t}\nonumber
    \left|\mathcal{Q}_{j}^{(t+1)}\right|=
    \begin{cases}
    \left|\mathcal{Q}_{j}^{(t)}\right|, & \text{if } j_t = 0  \\
    2\times \left|\mathcal{Q}_{j}^{(t)}\right|, & \text{if } j_t = 1
    \end{cases}
\end{equation}
by Lemma~\ref{lem:puncturing01}.
Thus, $|\mathcal{Q}_j^{(t+1)}|=\prod_{k=0}^{t} 2^{j_k}$, and accordingly,
\begin{equation}\label{eq:size_Q_n}\nonumber
    \left|\mathcal{Q}_{j}^{(n)}\right| = \prod_{t=0}^{n-1} 2^{j_t}=2^{\sum_{t=0}^{n-1}j_t} = 2^{d_\mathrm{H}(j)}.
\end{equation}
Setting $|\psi_j^{(0)}|=|\{\{j\}\}|=1$, we also have
\begin{equation}\label{eq:size_psi_t}\nonumber
    \left|\psi_{j}^{(t+1)}\right|=
    \begin{cases}
    2^{|\mathcal{Q}_j^{(t)}|} \times \left|\psi_{j}^{(t)}\right|, & \text{if } j_t = 0  \\
    \left|\psi_{j}^{(t)}\right|, & \text{if } j_t = 1
    \end{cases}
\end{equation}
by Lemma~\ref{lem:puncturing01}.
Hence,
\begin{equation}\label{eq:size_psi_n}\nonumber
    \left|\psi_{j}\right|=\prod_{t=0}^{n-1} 2^{\bar{j}_t \times |\mathcal{Q}_j^{(t)}| }
    = \prod_{t=0}^{n-1} 2^{\bar{j}_t \times \prod_{k=0}^{t-1} {2^{j_k }} }
    =2^{\sum_{t=0}^{n-1} \bar{j}_t \times \prod_{k=0}^{t-1} {2^{j_k }} }.
\end{equation}

The following theorem shows that there are two simple and important element sets in $\psi_j$ that comply with binary domination.

\begin{thm}\label{lem:puncturing02}
    For any $j\in \mathbb{Z}_{2^n}$, we have $\mathcal{D}_j \in \psi_j$ and $\bar{\mathcal{D}}_j\in\psi_j$, where $\bar{\mathcal{A}}=\{\bar{a}\mid a\in\mathcal{A}\}$ for $\mathcal{A}\subset \mathbb{Z}_{2^n}$.
\end{thm}
\begin{proof}
    See Appendix~\ref{proof:lem:puncturing02}.
\end{proof}

\if\doccolumn2
    \begin{table}[t!]
    \centering
    \caption{Min. Puncturing Bit Patterns for Each Incapable Bit $(N=8)$}
    \label{psi_table}
    \begin{tabularx}{\linewidth}{c X}
    \toprule
    Bit index $j$ & Minimal puncturing bit pattern $\psi_j^{(3)}$ \\
    \midrule
    $0=(000)_2$    &   $\mathbf{\{0\}}$, $\{1\}$, $\{2\}$, $\{3\}$, $\{4\}$, $\{5\}$, $\{6\}$, $\mathbf{\{7\}}$ \\
    $1=(001)_2$    &   $\mathbf{\{0,1\}}$, $\{2, 1\}$, $\{4, 1\}$, $\{6, 1\}$,\newline
                       ${\{0,3\}}$, $\{2, 3\}$, $\{4, 3\}$, $\{6, 3\}$,\newline
                       ${\{0,5\}}$, $\{2, 5\}$, $\{4, 5\}$, $\{6, 5\}$,\newline
                       ${\{0,7\}}$, $\{2, 7\}$, $\{4, 7\}$, $\mathbf{\{6, 7\}}$ \\
    $2=(010)_2$    &   $\mathbf{\{0,2\}}$, $\{0, 6\}$, $\{4, 2\}$, $\{4, 6\}$,\newline
                       ${\{1,3\}}$, $\{1, 7\}$, $\{5, 3\}$, $\mathbf{\{5, 7\}}$  \\
    $3=(011)_2$    &   $\mathbf{\{0, 1, 2, 3\}}$, $\{0, 1, 2, 7\}$, $\{0, 1, 6, 3\}$, $\{0, 1, 6, 7\}$,\newline
    $\{0, 5, 2, 3\}$, $\{0, 5, 2, 7\}$, $\{0, 5, 6, 3\}$, $\{0, 5, 6, 7\}$,\newline
    $\{4, 1, 2, 3\}$, $\{4, 1, 2, 7\}$, $\{4, 1, 6, 3\}$, $\{4, 1, 6, 7\}$,\newline
    $\{4, 5, 2, 3\}$, $\{4, 5, 2, 7\}$, $\{4, 5, 6, 3\}$, $\mathbf{\{4, 5, 6, 7\}}$\\
    $4=(100)_2$    &   $\mathbf{\{0,4\}}$, $\{1, 5\}$, $\{2, 6\}$, $\mathbf{\{3, 7\}}$ \\
    $5=(101)_2$    &   $\mathbf{\{0, 1, 4, 5\}}$, $\{2, 1, 5, 6\}$, $\{0, 3, 4, 7\}$, $\mathbf{\{2, 3, 6, 7\}}$ \\
    $6=(110)_2$    &   $\mathbf{\{0, 2, 4, 6\}}$, $\mathbf{\{1, 3, 5, 7\}}$ \\
    $7=(111)_2$    &   $\mathbf{\{0, 1, 2, 3, 4, 5, 6, 7\}}$ \\
    \bottomrule
    \end{tabularx}
    \end{table}
\else
    \begin{table}[t!]
    \centering
    \caption{Minimum Puncturing Bit Patterns for Each Incapable Bit $(N=8)$}
    \if\doccolumn1
    \vspace{-0.3cm}
    \fi
    \label{psi_table}
    \begin{tabularx}{\linewidth}{c X}
    \toprule
    Bit index $j$ & Minimal puncturing bit pattern $\psi_j^{(3)}$ \\
    \midrule
    $0=(000)_2$    &   $\mathbf{\{0\}}$, $\{1\}$, $\{2\}$, $\{3\}$, $\{4\}$, $\{5\}$, $\{6\}$, $\mathbf{\{7\}}$ \\
    $1=(001)_2$    &   $\mathbf{\{0,1\}}$, $\{2, 1\}$, $\{4, 1\}$, $\{6, 1\}$,
                       ${\{0,3\}}$, $\{2, 3\}$, $\{4, 3\}$, $\{6, 3\}$,\newline
                       ${\{0,5\}}$, $\{2, 5\}$, $\{4, 5\}$, $\{6, 5\}$,
                       ${\{0,7\}}$, $\{2, 7\}$, $\{4, 7\}$, $\mathbf{\{6, 7\}}$ \\
    $2=(010)_2$    &   $\mathbf{\{0,2\}}$, $\{0, 6\}$, $\{4, 2\}$, $\{4, 6\}$,
                       ${\{1,3\}}$, $\{1, 7\}$, $\{5, 3\}$, $\mathbf{\{5, 7\}}$  \\
    $3=(011)_2$    &   $\mathbf{\{0, 1, 2, 3\}}$, $\{0, 1, 2, 7\}$, $\{0, 1, 6, 3\}$, $\{0, 1, 6, 7\}$,
    $\{0, 5, 2, 3\}$, $\{0, 5, 2, 7\}$, $\{0, 5, 6, 3\}$, $\{0, 5, 6, 7\}$,\newline
    $\{4, 1, 2, 3\}$, $\{4, 1, 2, 7\}$, $\{4, 1, 6, 3\}$, $\{4, 1, 6, 7\}$,
    $\{4, 5, 2, 3\}$, $\{4, 5, 2, 7\}$, $\{4, 5, 6, 3\}$, $\mathbf{\{4, 5, 6, 7\}}$\\
    $4=(100)_2$    &   $\mathbf{\{0,4\}}$, $\{1, 5\}$, $\{2, 6\}$, $\mathbf{\{3, 7\}}$ \\
    $5=(101)_2$    &   $\mathbf{\{0, 1, 4, 5\}}$, $\{2, 1, 5, 6\}$, $\{0, 3, 4, 7\}$, $\mathbf{\{2, 3, 6, 7\}}$ \\
    $6=(110)_2$    &   $\mathbf{\{0, 2, 4, 6\}}$, $\mathbf{\{1, 3, 5, 7\}}$ \\
    $7=(111)_2$    &   $\mathbf{\{0, 1, 2, 3, 4, 5, 6, 7\}}$ \\
    \bottomrule
    \if\doccolumn1
    \vspace{-0.5cm}
    \fi
    \end{tabularx}
    \end{table}
\fi

When $N=8$, Table~\ref{psi_table} shows the minimal puncturing bit patterns that make each encoder input bit incapable.
They are determined by the recursive formula in Lemma~\ref{lem:puncturing01}.
Note that for $j\in\mathbb{Z}_8$, the bold-faced sets indicate $\mathcal{D}_j\in\psi_j$ and $\bar{\mathcal{D}}_j\in\psi_j$, which are given in Theorem~\ref{lem:puncturing02}.

Now we find all puncturing bit patterns inducing a given incapable bit pattern.
If two puncturing bit patterns give the same symmetric capacity for each bit channel,
they are said to be equivalent \cite{Chandesris2017}.
In this paper, we relax the equivalence concept.
If two puncturing bit patterns give the same incapable bit pattern, they are said to be \textit{widely equivalent}.
Based on Theorem~\ref{thm:incapable01} and Lemma~\ref{lem:puncturing01},
it is further possible to find all the widely equivalent puncturing bit patterns for a given incapable bit pattern.


Given two families $\psi$ and $\phi$ of sets, we write $\psi \vee \phi$ to denote the cross-union of these families, defined as
$\psi \vee \phi \triangleq \left\{ \mathcal{A} \cup \mathcal{B} \mid \mathcal{A}\in\psi, \mathcal{B}\in\phi \right\}$.
For a subset $\mathcal{A}\subset \mathbb{Z}_{2^n}$ such that $\mathcal{A}$ complies with binary domination,
let $\psi_{\mathcal{A}}$ denote the family of widely equivalent puncturing bit patterns making the bits with indices in $\mathcal{A}$ incapable.
We say that $j\in\mathcal{X}$ is a most dominant integer in $\mathcal{X}$ if there does not exist $\ell\in\mathcal{X}$ such that $j\prec \ell$.

\begin{thm}\label{thm:eq_punct}
Let $\mathcal{U}_p\subset\mathbb{Z}_{2^n}$ be an incapable bit pattern.
Then,
\begin{equation}\label{eq:eq_punct}\nonumber
\psi_{\mathcal{U}_p}=\left\{\mathcal{A} \in \bigvee_{j\in\breve{\mathcal{U}}_p} \psi_j \mathrel{\Big|} \left|\mathcal{A}\right| = \left| \mathcal{U}_p \right| \right\},
\end{equation}
where $\breve{\mathcal{U}}_p$ is the set of most dominant integers in $\mathcal{U}_p$.
\end{thm}
\begin{proof}
    See Appendix~\ref{proof:thm:eq_punct}.
\end{proof}


As an example, assume that $\mathcal{U}_p=\{0,1,2,4,5,6\}$ is given for a polar code of length $8$.
The set of most dominant integers in $\mathcal{U}_p$ is given as $\breve{\mathcal{U}}_p=\{5,6\}$.
Referring to Table~\ref{psi_table}, all the widely equivalent puncturing bit patterns for $\mathcal{U}_p$ are obtained as
$\{0,1,2,4,5,6\}$, $\{0,1,3,4,5,7\}$, $\{0,2,3,4,6,7\}$, and $\{1,2,3,5,6,7\}$.

Among the widely equivalent puncturing bit patterns obtained by Theorem~\ref{thm:eq_punct}, we are more interested in two simple patterns.
One is an \textit{identical puncturing bit pattern} in the sense that it is identical to a given incapable bit pattern,
as widely shown in the literature \cite{Eslami2011,Niu2013a,Shin2013,Zhang2014,Kim2015,Saber2015,Chandesris2017,Hong2017,Hong2018,El-Khamy2018}.
The following theorem formally shows that an identical puncturing bit pattern can be specified via binary domination.
Similar theorems and their proofs are given in \cite{Chandesris2017,Hong2018},
and readers may refer to them for comprehensive understanding.

\begin{thm} \label{thm:id_punct}
    Let $\mathcal{X}_p$ be the index set of $J$ punctured bits in $\mathbf{x}$ for a punctured polar code of length $M=2^n-J$,
    and let $\mathcal{U}_p$ be the index set of the incapable bits in $\mathbf{u}$ resulting from $\mathcal{X}_p$.
    Assume that $\mathcal{X}_p\leftarrow (\textbfit{P})_0^{J-1}$ for a $2^n$-posequence $\textbfit{P}$ such that $\mathcal{X}_p$ complies with binary domination.
    Then, $\mathcal{X}_p$ is an identical puncturing bit pattern, that is, $\mathcal{U}_p = \mathcal{X}_p$.
\end{thm}
\begin{proof}
    See Appendix~\ref{proof:thm:id_punct}.
\end{proof}

\begin{figure*}[t]
    \centering
  \subfloat[Identical puncturing bit pattern]{
       \includegraphics[width=0.48\linewidth]{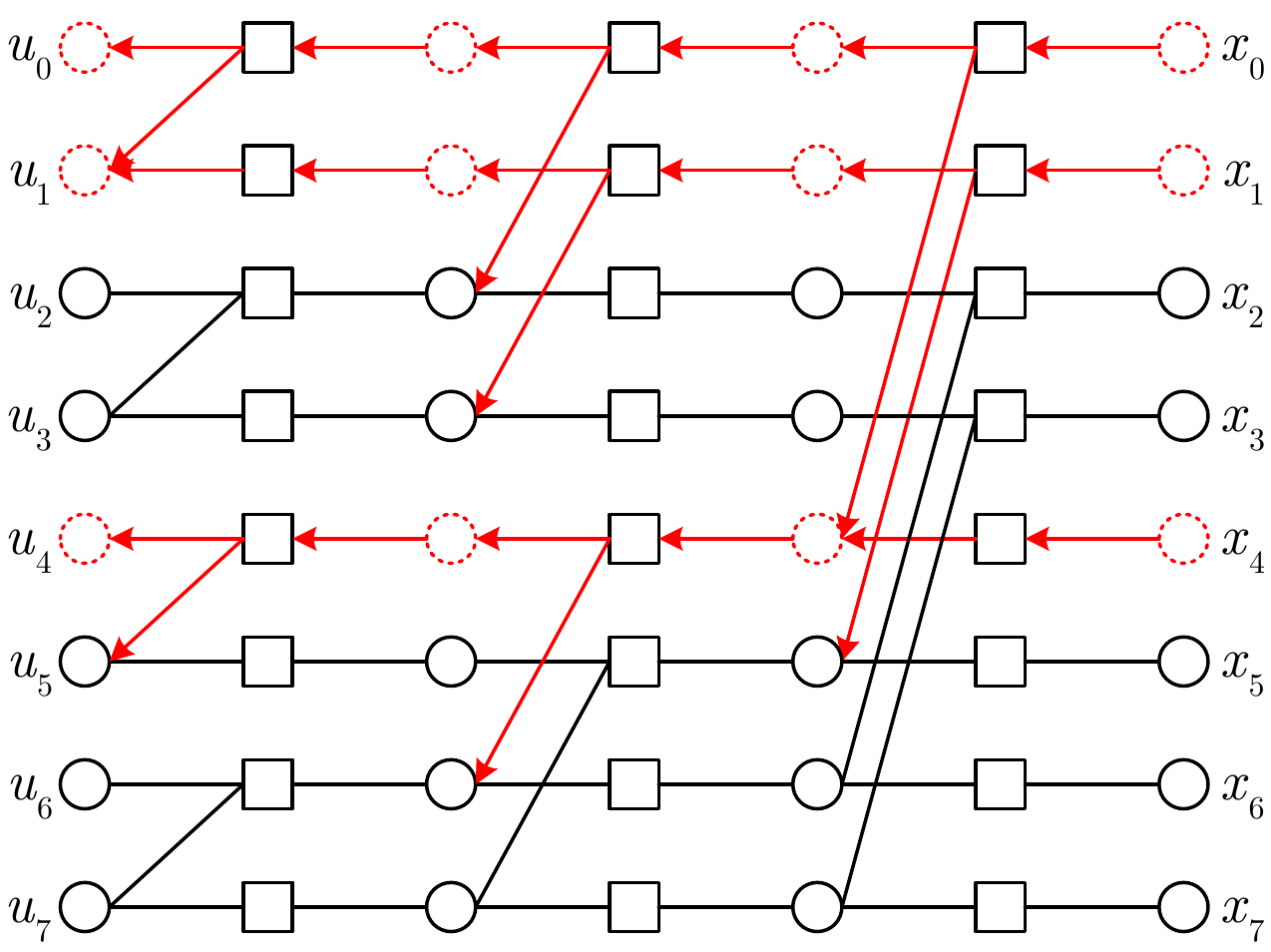}}
    \label{PO_punct}\hfill
  \subfloat[Reverse puncturing bit pattern]{
       \includegraphics[width=0.48\linewidth]{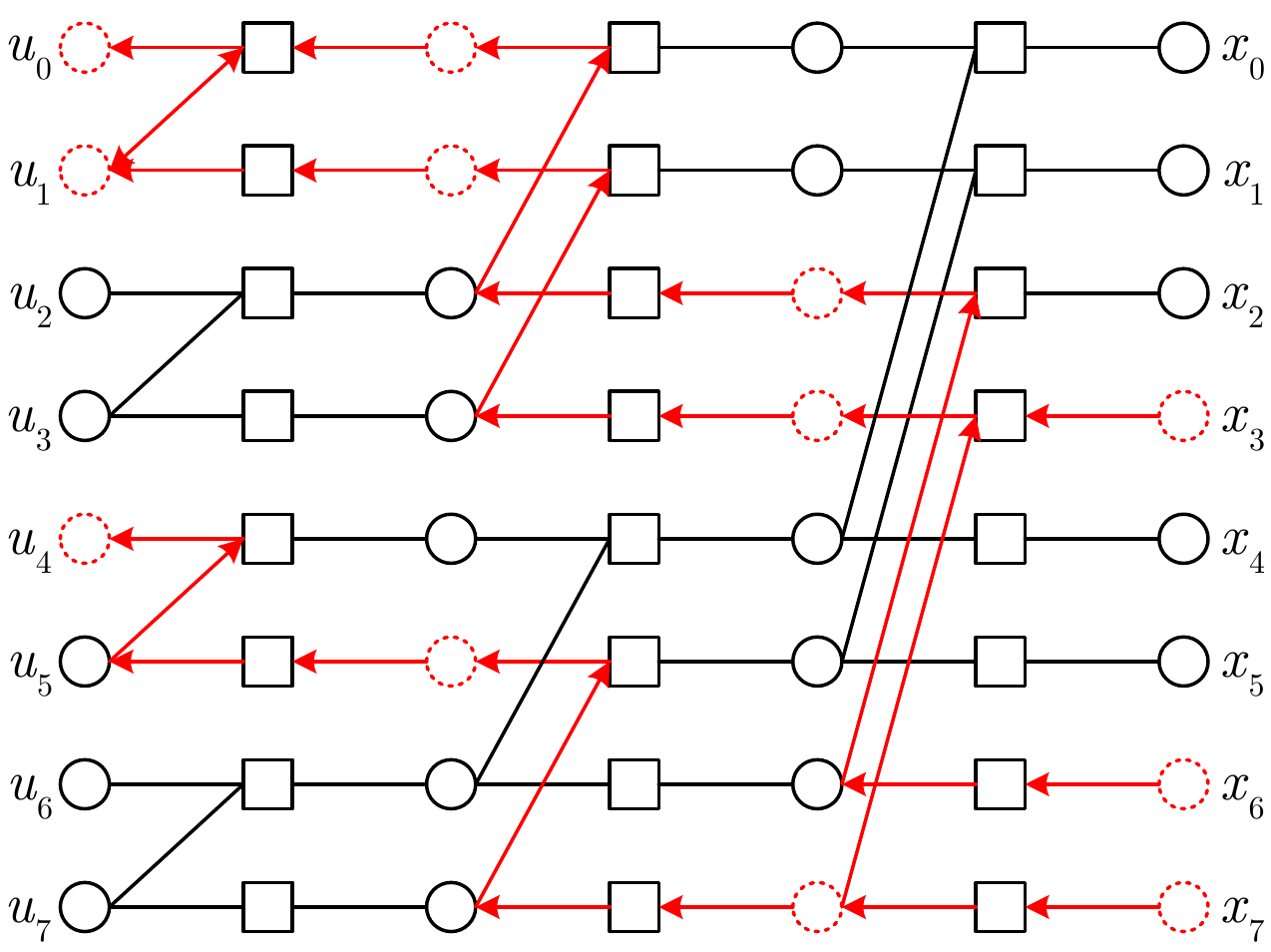}}
    \label{revPO_punct}\\
  \caption{Propagation of zero LLRs in a polar decoding graph of $N=8$.
  }
  \label{fig:punct}
  \if\doccolumn1
  \vspace{-0.5cm}
  \fi
\end{figure*}

An example of Theorem~\ref{thm:id_punct} is shown in Fig.~\ref{fig:punct} (a).
The puncturing bit pattern $\mathcal{P}=\{0,1,4\}$ complies with binary domination.
In the figure, the arrows indicate how the zero LLRs due to puncturing propagate, and the circles drawn by a dotted line correspond to the variable nodes with zero LLRs.
In each decoding stage, the zero LLR in a variable node is delivered to the variable node with the same index, since the puncturing bit pattern is subject to binary domination.
For example, $\alpha_4^{(3)}=0$ leads to $\alpha_4^{(2)}=0$, $\alpha_4^{(1)}=0$, and $\alpha_4^{(0)}=0$ in turn.

Most of the previous studies on puncturing considered identical puncturing bit patterns.
For example, Niu \textit{et al.} \cite{Niu2013a} proposed a quasi-uniform puncturing bit pattern, where the index set of $J$ punctured bits is $\{0,\ldots,J-1\}$ when the generator matrix $\mathbf{F}_2^{\otimes n}$ is considered.
Also, an identical subblock-wise permuted puncturing bit pattern has been adopted in the NR CB-RM \cite{TS38212}.
In these identical puncturing bit patterns, puncturing begins from low-indexed bits in $\mathbf{x}$ and the first punctured bit is $x_0$.

Another simple puncturing pattern is a \textit{reverse puncturing bit pattern} in the sense that it is the bitwise complement of the resultant incapable bit pattern.
In the following theorem, we show that a reverse puncturing bit pattern can also be identified via binary domination through the operation of SC decoding.

\begin{thm} \label{thm:rev_punct}
    Let $\mathcal{X}_p$ be the index set of $J$ punctured bits in $\mathbf{x}$ for a punctured polar code of length $M=2^n-J$,
    and let $\mathcal{U}_p$ be the index set of the incapable bits in $\mathbf{u}$ resulting from $\mathcal{X}_p$.
    Assume that $\mathcal{X}_p\leftarrow (\textbfit{P})_{2^n-J}^{2^n-1}$ for a $2^n$-posequence $\textbfit{P}$ such that $\mathcal{X}_p$ complies with binary domination.
    Then, $\mathcal{X}_p$ is a reverse puncturing bit pattern, that is, $\mathcal{U}_p = \bar{\mathcal{X}}_p$,
    where $\bar{\mathcal{X}}_p=\{\bar{x} \mid x\in\mathcal{X}_p\}$.
\end{thm}
\begin{proof}
    See Appendix~\ref{proof:thm:rev_punct}.
\end{proof}

Fig.~\ref{fig:punct} (b) gives an example of Theorem~\ref{thm:rev_punct}.
The puncturing bit pattern $\mathcal{X}_p=\{7,6,3\}$ complies with binary domination.
It brings about the bitwise-complemented incapable bit pattern,
where puncturing $x_i$ makes $u_{\bar{i}}=u_{7-i}$ incapable.
For example, given that $x_7$ and $x_6$ are punctured,
$\alpha_3^{(3)}=0$ leads to $\alpha_7^{(2)}=0$, $\alpha_5^{(1)}=0$, and $\alpha_4^{(0)}=0$ in turn.
Finally, we have $\mathcal{U}_p=\{\bar{7},\bar{6},\bar{3}\}=\{0,1,4\}$.

Based on Theorem~\ref{thm:rev_punct}, a puncturing bit pattern can be determined by the bitwise complement of a desired incapable bit pattern.
Since the incapable bit pattern is constrained by binary domination,
the corresponding reverse puncturing bit pattern also complies with binary domination.
This enables us to begin puncturing from high-indexed bits in $\mathbf{x}$, in contrast to conventional identical puncturing bit patterns.
In the reverse puncturing bit pattern, the first punctured bit is $x_{2^n-1}$.

\section{Shortening and Fixed Bit Patterns}\label{sec:short}

The relation between a shortening bit pattern and its corresponding fixed bit pattern can be explicitly explained by the encoding procedure in \eqref{eq:polar_enc}.
For simplicity, we use $\mathbf{A}_{i,j}$ to denote the entry at row $i$ and column $j$ of a matrix $\mathbf{A}$.
Given an index set $\mathcal{A}$, we write $(\mathbf{A})_{\mathcal{A}}$ to denote the submatrix of a matrix $\mathbf{A}$ formed by the rows with indices in $\mathcal{A}$.


Shortening a code is a modification method to reduce its dimension and length by a given number.
The basic idea of shortening a polar code is to fix the values of $J$ bits in $\mathbf{u}$ to zero so that the values of $J$ bits in $\mathbf{x}$ also become zero.
Let $\mathcal{U}_s$ be the index set of $J$ shortened bits, and let $\mathcal{U}_s^\mathsf{c}=\mathbb{Z}_{2^n}\backslash \mathcal{U}_s$.
Then, every polar codeword of length $2^n$ can be expressed as
\begin{equation}\label{eq:short_1}\nonumber
\mathbf{x} = \mathbf{u}_{\mathcal{U}_s^\mathsf{c}} \left(\mathbf{F}_2^{\otimes n}\right)_{\mathcal{U}_s^\mathsf{c}} + \mathbf{u}_{\mathcal{U}_s} \left(\mathbf{F}_2^{\otimes n}\right)_{\mathcal{U}_s}.
\end{equation}
By letting $\mathbf{u}_{\mathcal{U}_s}=\mathbf{0}$ for shortening, we have
\begin{equation}\label{eq:short_2}\nonumber
\mathbf{x} = \mathbf{u}_{\mathcal{U}_s^\mathsf{c}} \left(\mathbf{F}_2^{\otimes n}\right)_{\mathcal{U}_s^\mathsf{c}}.
\end{equation}
Note that $(\mathbf{F}_2^{\otimes n})_{\mathcal{U}_s^\mathsf{c}} \in \mathbb{F}_2^{(N-J) \times N}$ is the effective generator matrix obtained from shortening $\mathbf{u}_{\mathcal{U}_s}$.
In order to fix the values of $J$ bits in $\mathbf{x}$ regardless of $\mathbf{u}_{\mathcal{U}_s^\mathsf{c}}$,
all the entries at certain $J$ columns in $(\mathbf{F}_2^{\otimes n})_{\mathcal{U}_s^\mathsf{c}}$ should be zero.



Letting $\mathcal{X}_s$ denote the index set of the encoder output bits fixed to zero by shortening,
we first show in the following theorem that $\mathcal{X}_s$ complies with binary domination.

\begin{thm}\label{thm:fixed}
    For any $j\in\mathbb{Z}_{2^n}$,
    if $j\in\mathcal{X}_s$, then $\mathcal{G}_j\subseteq\mathcal{X}_s$.
    In other words, $\mathcal{X}_s$ complies with binary domination.
\end{thm}
\begin{proof}
Recall \cite{Sarkis2016} that the entry at row $i$ and column $j$ of $\mathbf{F}_2^{\otimes n}$ is given by
\begin{equation}\label{eq:one_pos}
    \left(\mathbf{F}_2^{\otimes n}\right)_{i,j}=
    \begin{cases}
    1, & \text{if } i \succeq j \\
    0, & \text{otherwise,}
    \end{cases}
\end{equation}
and from \eqref{eq:polar_enc} and \eqref{eq:one_pos}, we have
\begin{equation}\label{eq:one_bit_encode}
    x_j = \sum_{i\in \mathcal{G}_j} u_i.
\end{equation}
Hence, $u_i$ needs to be shortened for all $i\in\mathcal{G}_j$ in order to make $x_j$ fixed to zero regardless of the encoder input vector $\mathbf{u}$.
That is, $\mathcal{G}_j\subset\mathcal{U}_s$ if and only if $j\in\mathcal{X}_s$.
Assume that $j\in\mathcal{X}_s$, so $\mathcal{G}_j\subset\mathcal{U}_s$.
Then, for any $k\in\mathcal{G}_j$, we have $\mathcal{G}_k\subset\mathcal{U}_s$ which results in $k\in\mathcal{X}_s$.
Therefore, we have $\mathcal{G}_j \subset \mathcal{X}_s$.
\end{proof}


\begin{cor}\label{cor:short}
    For any shortened polar code, we have $\mathcal{U}_s=\mathcal{X}_s$.
\end{cor}
\begin{proof}
    Let $\mathcal{S}_j$ denote the index set of the shortened bits required to make $x_j$ fixed.
    Then, we have $\mathcal{S}_j = \mathcal{G}_j$ from \eqref{eq:one_bit_encode}.
    Thus, we have $\mathcal{U}_s=\cup_{j\in\mathcal{X}_s} \mathcal{S}_j = \cup_{j\in\mathcal{X}_s} \mathcal{G}_j = \mathcal{X}_s$, where the last equality comes from Theorem~\ref{thm:fixed}.
\end{proof}

One consequence of Corollary~\ref{cor:short} is that any feasible shortening bit pattern complies with binary domination.
Clearly, the greedy selection method with the single-weight column criterion in \cite{Wang2014} generates a shortening bit pattern following the partial order $\succ$, thereby resulting in $\mathcal{X}_s=\mathcal{U}_s$.
We further identify a necessary and sufficient condition for an encoder output bit to be fixed by additionally shortening a single encoder input bit in $\mathbf{u}$ in the following theorem.

\begin{thm}\label{thm:short}
    Let $j\in\mathbb{Z}_{2^n} \backslash \mathcal{X}_s$.
    The encoder output bit $x_j$ can be fixed to zero by additionally shortening a single encoder input bit if and only if $\hat{\mathcal{G}}_j \subset \mathcal{X}_s$.
\end{thm}
\begin{proof}
    If $x_j$ is additionally made fixed, then $\mathcal{G}_j\subset \mathcal{X}_s \cup \{j\}$ by Theorem~\ref{thm:fixed},
    that is, $\hat{\mathcal{G}}_j \subset \mathcal{X}_s$.
    To prove the converse, assume that $\hat{\mathcal{G}}_j\subset\mathcal{X}_s$.
    Then, $\hat{\mathcal{G}}_j\subset\mathcal{U}_s$ since $\mathcal{X}_s=\mathcal{U}_s$ by Corollary~\ref{cor:short}.
    Combining this relation with \eqref{eq:one_bit_encode}, we have $x_j = \sum_{i\in\hat{\mathcal{G}}_j}u_i + u_j=u_j$.
    Hence, $x_j$ can be made fixed by shortening $u_j$.
\end{proof}


\section{Design of Unified Circular-Buffer Rate Matching}\label{sec:unified}

\if\doccolumn2
    \begin{table*}[t!]
    \centering
    \caption{Summary of Constraints for Polar Code Rate-Matching Based on Binary Domination}
    \label{rm_constraint}
    \begin{tabularx}{\linewidth}{c X X}
    \toprule
    Rate matching  & Encoder input $\mathbf{u}$ & Encoder output $\mathbf{x}$ \\
    \midrule
    $J$-bit puncturing  & first $J$ components of a $2^n$-posequence \newline (Theorems~\ref{thm:incapable01}~and~~\ref{thm:incapable02}) & widely equivalent puncturing bit patterns (Theorem~\ref{thm:eq_punct}) including \newline
    $\cdot$ an identical pattern to incapable pattern (Theorem~\ref{thm:id_punct},\cite{Chandesris2017,Hong2018}) \newline
    $\cdot$ a reverse pattern to incapable pattern (Theorem~\ref{thm:rev_punct}) \\
    $J$-bit shortening  & an identical pattern to its fixed bit pattern (Corollary~\ref{cor:short}) & last $J$ components of a $2^n$-posequence (Theorems~\ref{thm:fixed} and \ref{thm:short}) \\
    \bottomrule
    \end{tabularx}
    \end{table*}
\else
\begin{table*}[t!]
    \centering
    \caption{Summary of Constraints for Polar Code Rate-Matching Based on Binary Domination}
    \label{rm_constraint}
    \begin{tabularx}{\linewidth}{c X X}
    \toprule
    Rate matching  & Encoder input $\mathbf{u}$ & Encoder output $\mathbf{x}$ \\
    \midrule
    $J$-bit puncturing  & first $J$ components of a $2^n$-posequence \newline (Theorems~\ref{thm:incapable01}~and~~\ref{thm:incapable02}) & widely equivalent puncturing bit patterns (Theorem~\ref{thm:eq_punct}) including \newline
    $~~~$- an identical pattern to its incapable bit pattern \newline $~~~~$ (Theorem~\ref{thm:id_punct}) \newline
    $~~~$- a reverse pattern to its incapable bit pattern \newline $~~~~$ (Theorem~\ref{thm:rev_punct}) \\
    $J$-bit shortening  & an identical pattern to its fixed bit pattern \newline (Corollary~\ref{cor:short}) & last $J$ components of a $2^n$-posequence \newline (Theorems~\ref{thm:fixed} and \ref{thm:short}) \\
    \bottomrule
    \end{tabularx}
    \end{table*}
\fi

\subsection{Unified Rate Matching Bit Pattern}

As shown in the previous sections, both incapable and shortening bit patterns at the encoder input comply with binary domination.
Table~\ref{rm_constraint} summarizes how polar rate-matching patterns are determined.
The shortening bit patterns designed by the single-weight column criterion are restricted to comply with binary domination.
In this setting, a $J$-bit shortening bit pattern is determined by the last $J$ components of a $2^n$-posequence
and is identical to its corresponding fixed bit pattern.
Incapable bit patterns are also strictly constrained by binary domination so that a $J$-bit incapable bit pattern is given by the first $J$ components of a $2^n$-posequence.
There are multiple widely-equivalent puncturing bit patterns that result in the same incapable bit pattern,
so we have a degree of freedom to choose a puncturing bit pattern.
Among these widely equivalent patterns, the reverse puncturing bit pattern is a special one whose bit indices are obtained by the bitwise complement of the given incapable bit pattern.

Consider a practical CB-RM scheme using a single nested sequence
such that the index sets of $J$ puncturing and shortened bits at the encoder output are determined by selecting the first $J$ and the last $J$ elements of the given sequence, respectively.
According to Table~\ref{rm_constraint},
the sequence should be designed to be a $2^n$-posequence,
and any $2^n$-posequence can be a candidate for it.
Thus, a unified rate matching bit pattern of length $2^n$ can be optimized by finding the best one among all $2^n$-posequences in terms of the rate-compatible performance.

\if\doccolumn2
    \begin{table}[h]
    \centering
    \caption{Size of Search Space for Optimizing $2^n$-posequences}
    \label{num_poseq}
    \begin{tabular}{c c}
    \toprule
    Polar code size $2^n$ & Number of all possible $2^n$-posequences \\
    \midrule
    2    &   $1$ \\
    4    &   $2~(=2^1)$\\
    8    &   $48~(=2^4 \times 3^1)$\\
    16    &  $1,680,384~(=2^{10} \times 3^1 \times 547)$\\
    \bottomrule
    \end{tabular}
    \end{table}
\else
    {\renewcommand{\arraystretch}{0.8}
    \begin{table}[t]
    \centering
    \vspace{-0.3cm}
    \caption{Size of Search Space for Optimizing $2^n$-posequences}
    \vspace{-0.3cm}
    \label{num_poseq}
    \begin{tabular}{c c}
    \toprule
    Polar code size $2^n$ & Number of all $2^n$-posequences \\
    \midrule
    2    &   $1$ \\
    4    &   $2~(=2^1)$\\
    8    &   $48~(=2^4 \times 3^1)$\\
    16    &  $1,680,384~(=2^{10} \times 3^1 \times 547)$\\
    \bottomrule
    \vspace{-1.0cm}
    \end{tabular}
    \end{table}
    }
\fi

Table~\ref{num_poseq} shows the number of $2^n$-posequences obtained by an exhaustive computer search.
Even in the case that $2^n=16$, the search space for optimizing a rate matching bit pattern is very large, approximately $1.68\times 10^{6}$.
As $2^n$ increases, the size of the search space becomes prohibitively large.
Therefore, it is an interesting problem in practical polar code construction to efficiently optimize a unified rate matching bit pattern in a reduced search space.

\subsection{Unified Circular-Buffer Rate Matching}

Based on the observations given in Table~\ref{rm_constraint}, we propose a practical unified CB-RM scheme, in which both a fixed bit pattern and a puncturing bit pattern at the encoder output are aligned in identical order.
Fig.~\ref{fig:unified_RM} describes a polar coding chain with unified CB-RM.
Given code parameters $N$, $M$ and $R=K/M$,
a rate-matching technique to be applied is determined in advance before encoding.
If $M>N$, then repetition is applied.
Otherwise, shortening is usually employed for high code rates, while puncturing is configured for low code rates, as shown in the NR polar coding chain.
Then, split channel allocation is carried out, depending on the determined rate-matching technique, and the linear transformation is finally performed.

\begin{figure}[t]
	\centering
    \if\doccolumn1
    \vspace{-0.3cm}
    \fi
    \includegraphics[width=8.85cm]{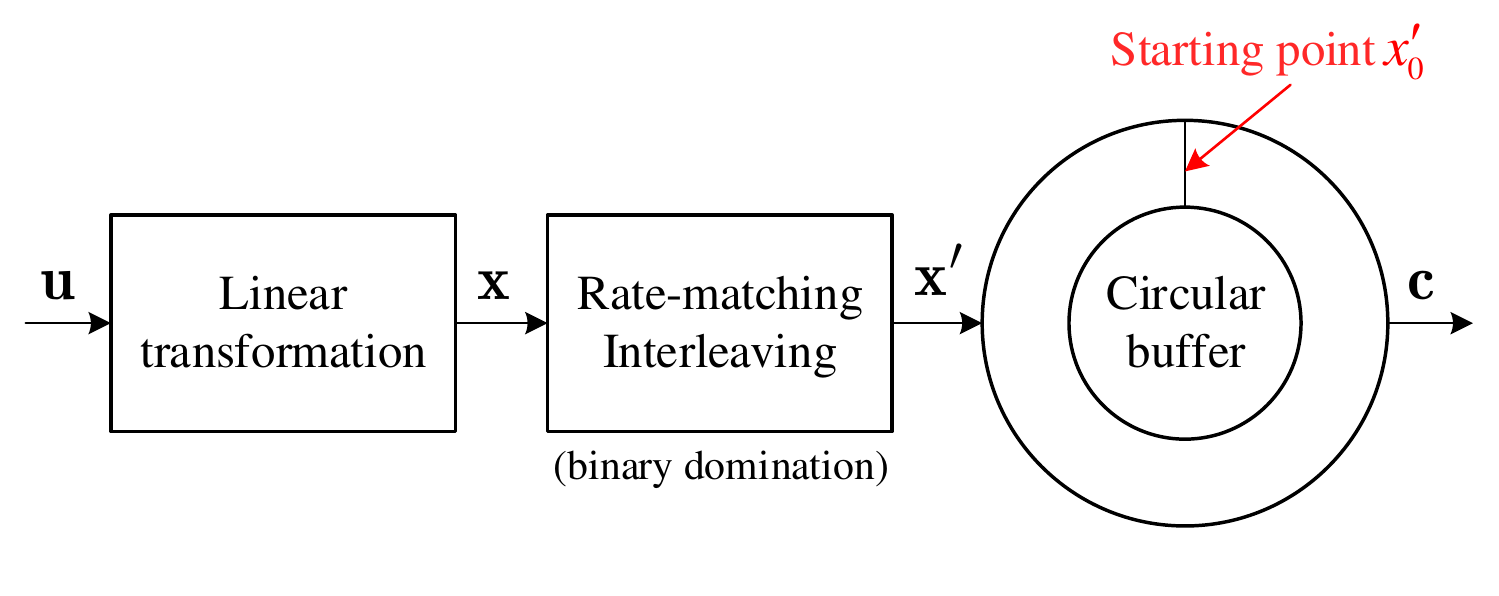}
    \if\doccolumn1
    \vspace{-0.5cm}
    \fi
	\caption{Block diagram of the proposed circular-buffer rate-matching scheme.
    A rate-matching interleaver is designed so that the interleaving pattern complies with binary domination.
    The starting point of the circular buffer is always set to zero for all rate-matching techniques: puncturing, shortening, and repetition.
}
	\label{fig:unified_RM}
    \if\doccolumn1
    \vspace{-0.5cm}
    \fi
\end{figure}

In the proposed CB-RM scheme, a resulting encoding output sequence $\mathbf{x}$ is interleaved in a predetermined order.
Binary domination is the only constraint that we have in the design of a rate-matching interleaver.
Let $\mathbf{x}^\prime=(x^\prime_0,\ldots,x^\prime_{N-1})$ be the output sequence of the rate-matching interleaver corresponding to $\mathbf{x}$.
After interleaving, $\mathbf{x}^\prime$ is stored in a circular buffer, and a desired codeword $\mathbf{c}$ is obtained by extracting $M$ bits in order from the buffer.
The starting point for bit selection from the buffer always indicates $x^\prime_0$, regardless of the employed rate-matching technique.
When either puncturing or shortening is configured,
the bits $x^{\prime}_0,\ldots,x^{\prime}_{M-1}$ are transmitted, while the bits $x^\prime_M,\ldots,x^\prime_{N-1}$ are discarded.
For repetition, $M-N$ bits are additionally selected in the circular order so that they are repeated.

\begin{figure*}[t]
    \centering
    \if\doccolumn1
    \vspace{-0.3cm}
    \fi
  \subfloat[Puncturing four bits from the tail of the buffer]{
       \includegraphics[width=0.45\linewidth]{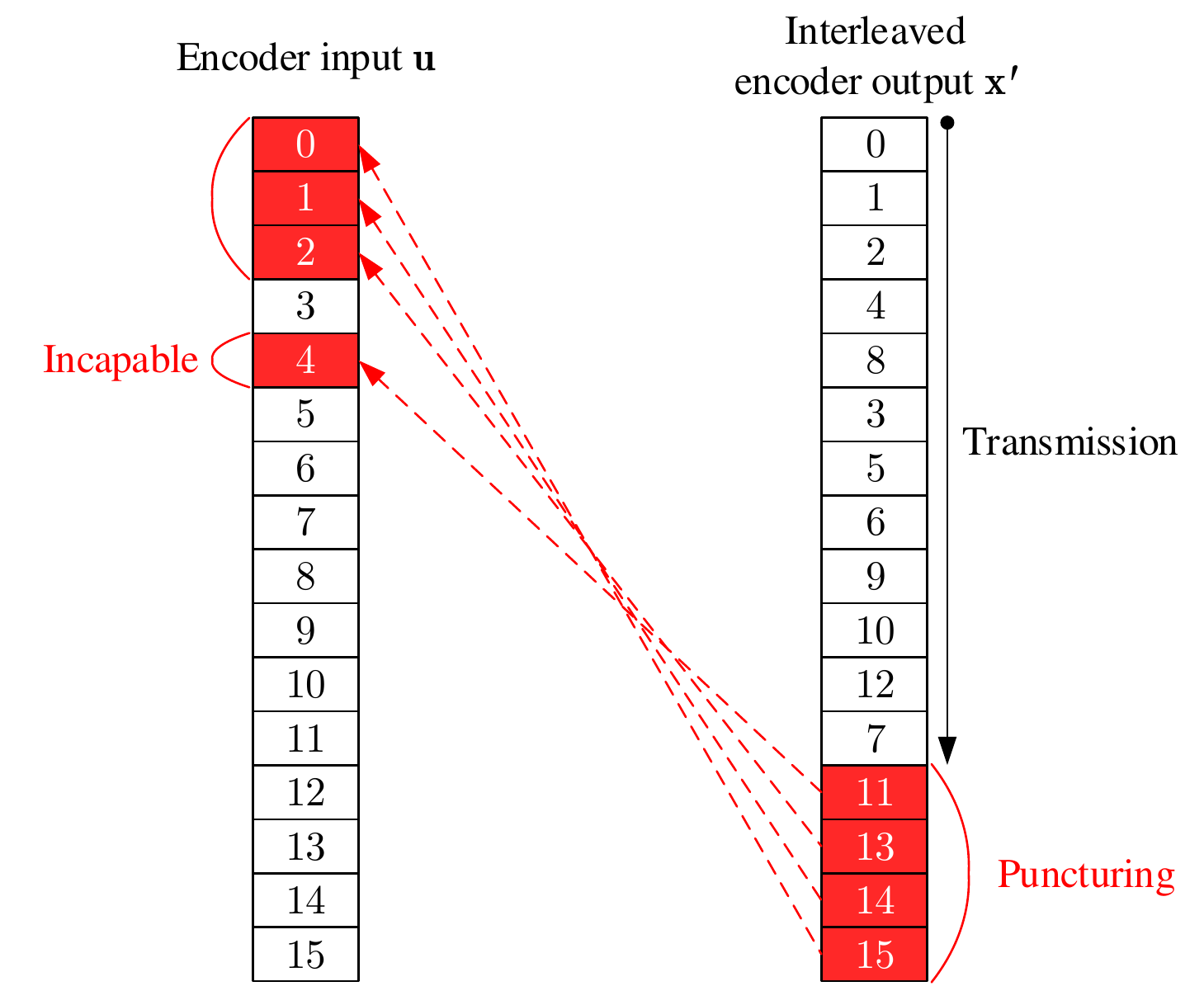}}
    \label{URM_punct}\hfill
  \subfloat[Shortening seven bits from the tail of the buffer]{
       \includegraphics[width=0.45\linewidth]{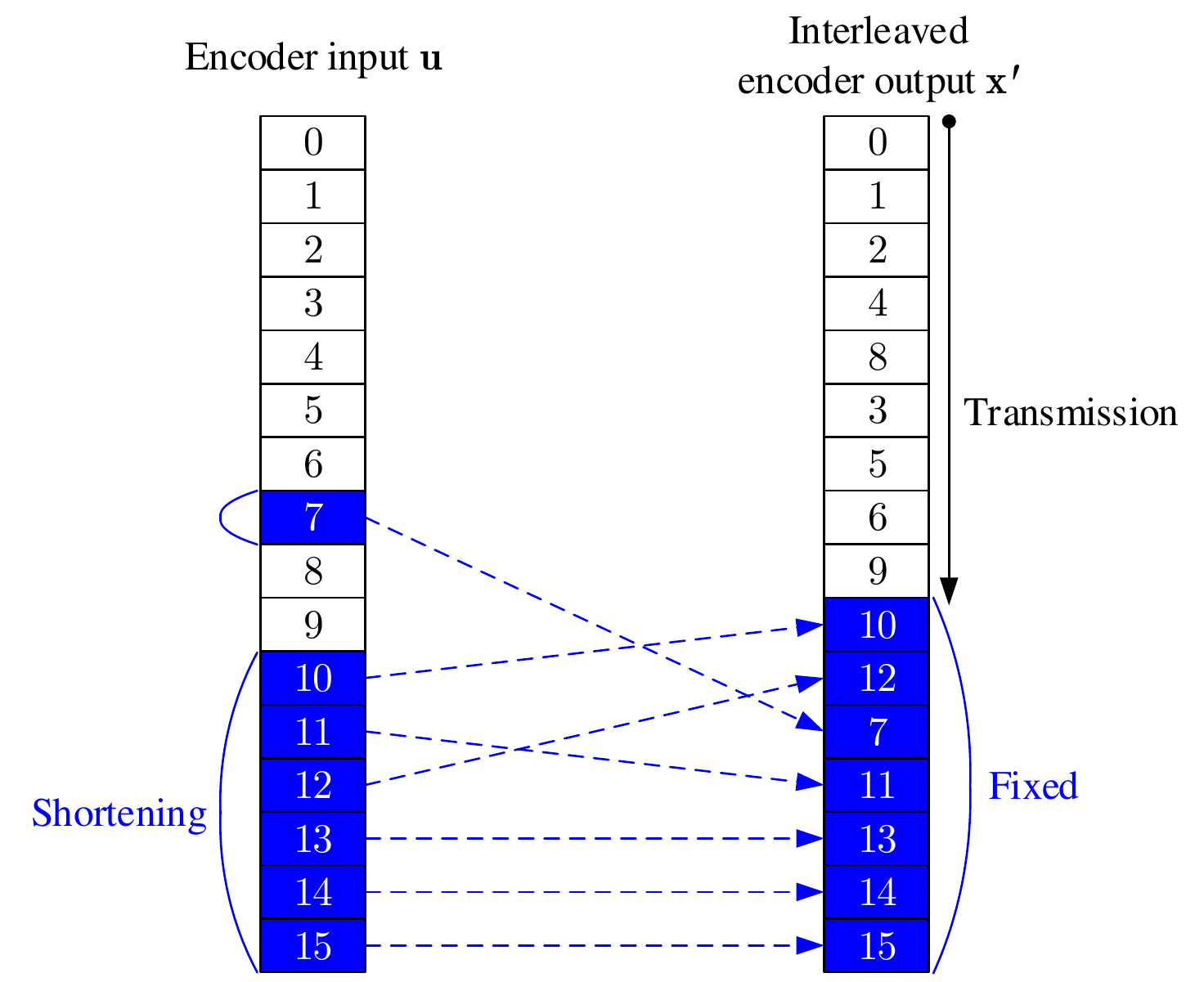}}
    \label{URM_short}\\
  \caption{Example of unified rate matching and corresponding split channel allocation.
  Encoder output bit sequence $\mathbf{x}$ is interleaved by a predetermined pattern $(0,1,2,4,8,3,5,6,9,10,12,7,11,13,14,15)$,
  and bit extraction starts from $x^\prime_0$ for all rate-matching techniques: puncturing, shortening, and repetition.
  In 4-bit puncturing, $x_{15}, x_{14}, x_{13}$, and $x_{11}$ are punctured, and thus, $u_{0}, u_{1}, u_{2}$, and $u_{4}$ are made incapable by a reverse puncturing bit pattern.
  In 7-bit shortening, $u_{15}, u_{14}, u_{13}, u_{11}, u_{7}, u_{12}$, and $u_{10}$ are shortened in order to make $x_{15}, x_{14}, x_{13}, x_{11}, x_{7}, x_{12}$, and $x_{10}$, which are to be not transmitted, have deterministic values.
  }
  \label{fig:URM}
  \if\doccolumn1
  \fi
\end{figure*}

The proposed CB-RM scheme is a \textit{unified} rate-matching scheme in the sense that the rate-matching interleaving and the circular buffer management are always the same for puncturing, shortening, and repetition.
The only thing that we need to care about is the split channel allocation.
Fig.~\ref{fig:URM} describes a simple example of the split channel allocation for puncturing and shortening when the proposed CB-RM scheme is applied.
In this example, the unified rate-matching bit pattern for $N=16$ is given as $(0,1,2,4,8,3,5,6,9,10,12,7,11,13,14,15)$, which is a $16$-posequence.
Assuming that puncturing is configured for $M=12$,
the bits $x_{15}$, $x_{14}$, $x_{13}$, and $x_{11}$ are punctured by the proposed CB-RM scheme.
Since the bits $u_0$, $u_1$, $u_2$, and $u_4$ are then made incapable, they are excluded from the allocation of information bits.
On the other hand, when shortening is applied for $M=9$,
the bits $u_7$, $u_{10}$, $u_{11}$, $u_{12}$, $u_{13}$, $u_{14}$ and $u_{15}$ are shortened in order to fix the values of encoder output bits with the same indices, \textit{i.e.}, $x_{7}$, $x_{10}$, $x_{11}$, $x_{12}$, $x_{13}$, $x_{14}$, and $x_{15}$.

\if\doccolumn2
\begin{figure}[h]
	\centering
    \if\doccolumn1
    \vspace{-0.3cm}
    \fi
    \includegraphics[width=8.5cm]{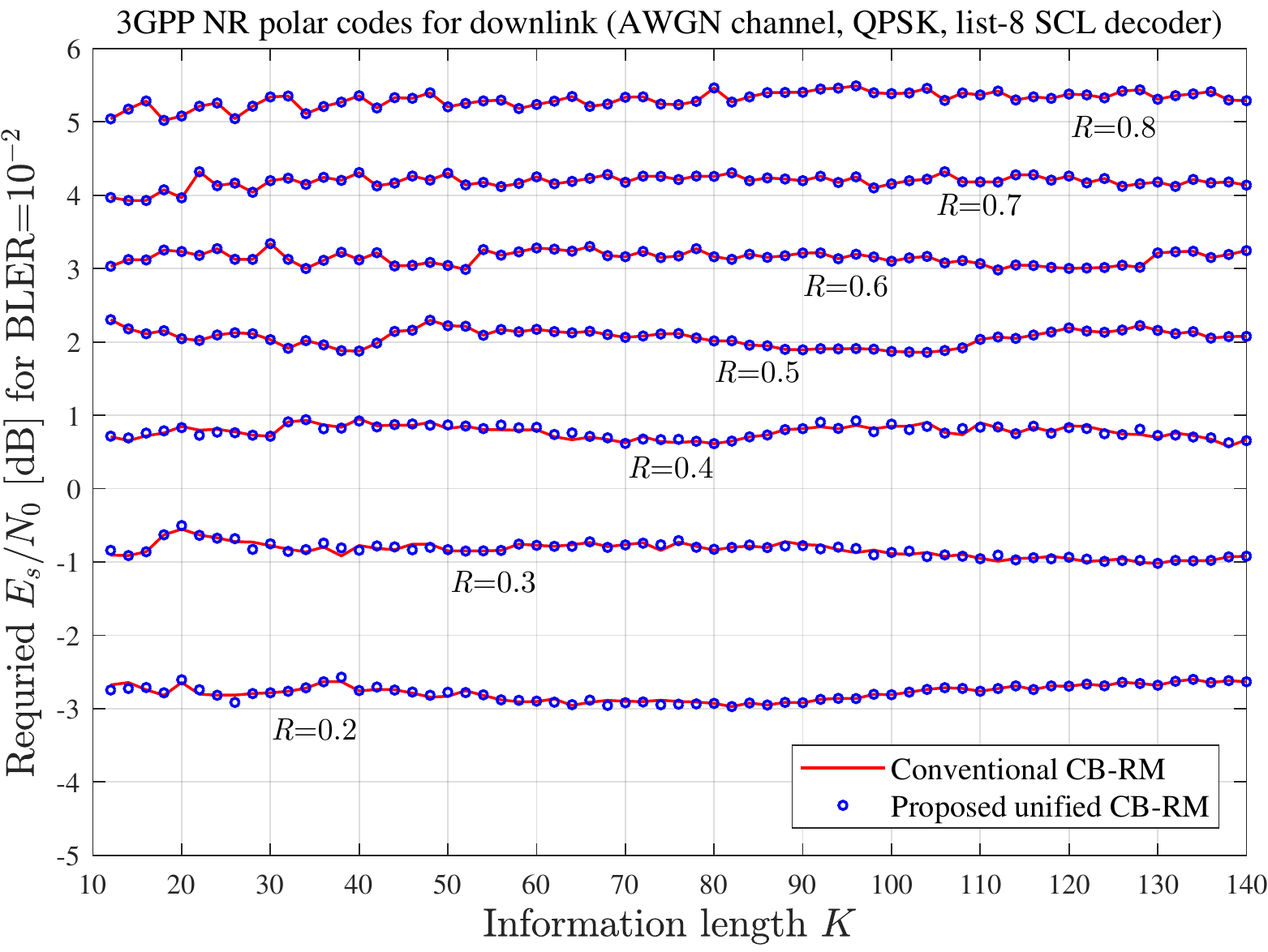}
    \if\doccolumn1
    \vspace{-0.5cm}
    \fi
	\caption{Performance of 3GPP NR downlink polar codes with conventional CB-RM and proposed CB-RM.
}
	\label{fig:perf_eval}
    \if\doccolumn1
    \vspace{-0.5cm}
    \fi
\end{figure}
\else
\begin{figure}[t]
	\centering
    \if\doccolumn1
    \vspace{-0.3cm}
    \fi
    \includegraphics[width=8.5cm]{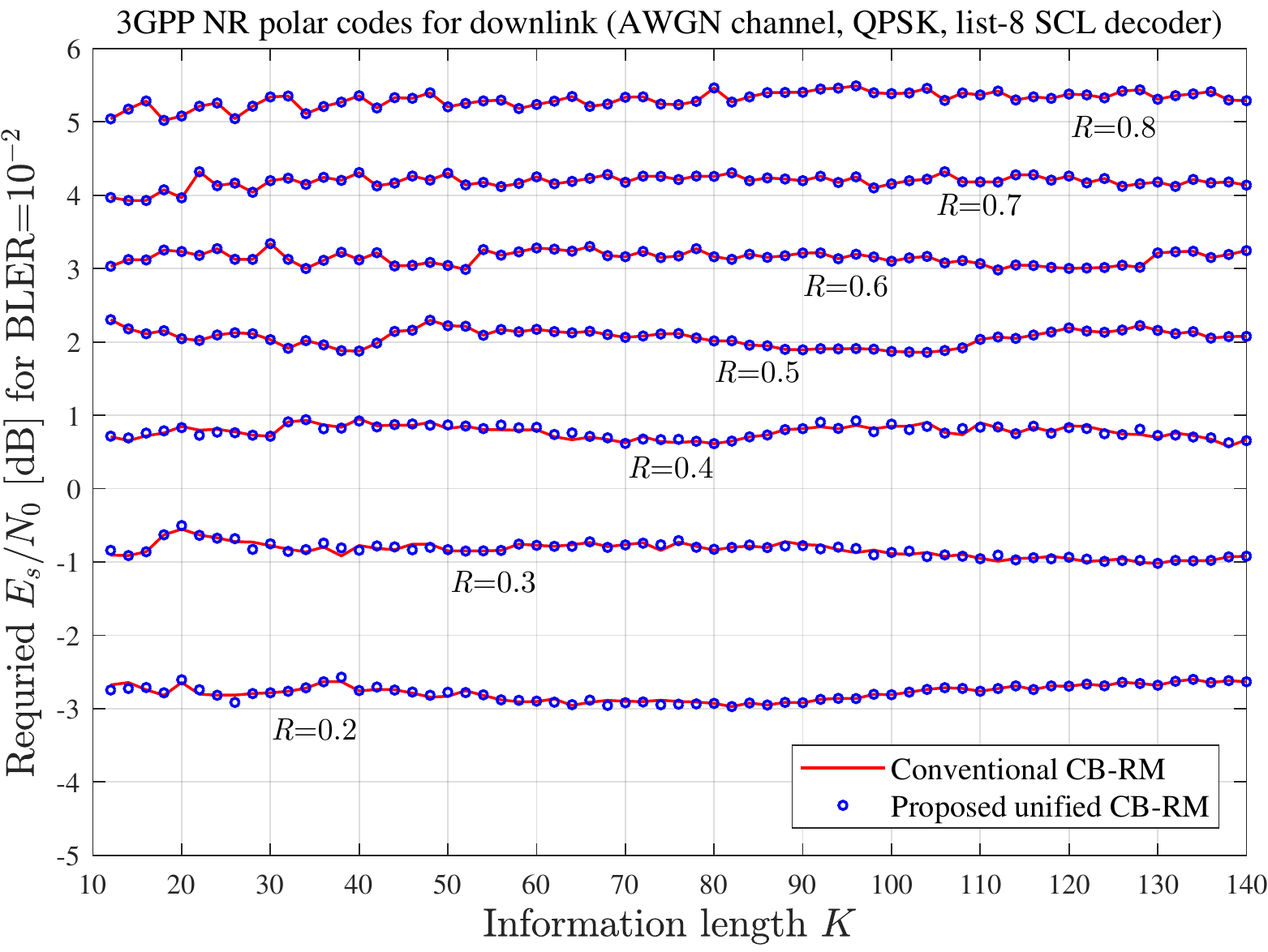}
    \if\doccolumn1
    \vspace{-0.5cm}
    \fi
	\caption{Performance of 3GPP NR downlink polar codes with conventional CB-RM and proposed CB-RM.
}
	\label{fig:perf_eval}
    \if\doccolumn1
    \vspace{-0.5cm}
    \fi
\end{figure}
\fi
In most of previous studies, the interleaving pattern and the order for puncturing are different from those for shortening.
For a practical example,
the rate-matching scheme of the NR polar coding system \cite{TS38212} uses a unified interleaver based on subblock-wise permutation,
but extracts bits from the buffer that vary with the employed rate-matching technique.
In fact, the starting point of the circular buffer is set to $x^\prime_{N-J}$ for puncturing,
whereas it points to $x^\prime_0$ for shortening and repetition.
On the other hand, the proposed CB-RM scheme can be more simply implemented for a practical coding chain.
The starting point of the NR circular-buffer in the proposed scheme can always be set to $x^\prime_0$  without any further changes, regardless of the rate-matching technique employed.
As a result, the NR coding chain with the proposed CB-RM scheme can be simpler and more efficient, while the average rate-matching performance remains intact under the same code configurations and evaluation conditions, as shown in Fig.~\ref{fig:perf_eval}.

\section{Conclusion and Future Works}\label{sec:conclude}

Due to the fixed structure of encoding and the behaviour of SC-based decoding for polar codes,
puncturing and shortening should be carefully performed in a certain order.
We showed that binary domination completely determines the incapable and shortening bit patterns for polar codes.
Based on this observation, we proposed a unified rate-matching scheme to support simple and efficient rate matching for polar codes.

Several interesting problems on rate matching for polar codes remain unsolved.
One of the most important issues is to optimize the incapable and shortening bit patterns.
In order to do this, we need to find an efficient way of reducing the search space for rate-matching patterns.
As shown in Table~\ref{num_poseq}, the number of all the possible $2^n$-posequences is too large even in the case of $N=16$, so it is impractical to find the best one among them.
The search space may be reduced by properly adding some constrains such as the implementation complexity for practical applications.
However, this reduction should be made without any significant loss of the optimized performance.


\ifCLASSOPTIONcaptionsoff
  \newpage
\fi

\appendix

\subsection{Proof of Lemma~\ref{lem:incapable02}}\label{proof:lem:incapable02}

Consider two variable nodes $v_{i}^{(t)}$ and $v_{i+2^t}^{(t)}$ for $i\in 2\ell\cdot 2^t+\mathbb{Z}_{2^t}$.
In the binary representation, $i_t=0$ and $\left(i+2^t\right)_{t}=1$.
Hence, the LLRs of these two variable nodes, $\alpha_i^{(t)}$ and $\alpha_{i+2^t}^{(t)}$, are calculated by \eqref{eq:f_func} and
\eqref{eq:g_func} as
$\alpha_i^{(t)}=f\big{(}\alpha_i^{(t+1)}, \alpha_{i+2^t}^{(t+1)}\big{)}$
and
$\alpha_{i+2^t}^{(t)}=g\big{(}\alpha_i^{(t+1)}, \alpha_{i+2^t}^{(t+1)}, \beta_i^{(t)}\big{)}$,
respectively.
It is easily shown that $\alpha_{i}^{(t)}=0$ if $\alpha_{i+2^t}^{(t)}=0$ by checking all the four cases:
$\big{(}\alpha_i^{(t+1)}=0,~\alpha_{i+2^t}^{(t+1)}=0\big{)}$,
$\big{(}\big{|}\alpha_i^{(t+1)}\big{|}>0,~\alpha_{i+2^t}^{(t+1)}=0\big{)}$,
$\big{(}\alpha_i^{(t+1)}=0,~\big{|}\alpha_{i+2^t}^{(t+1)}\big{|}>0\big{)}$,
and $\big{(}\big{|}\alpha_i^{(t+1)}\big{|}>0,~\big{|}\alpha_{i+2^t}^{(t+1)}\big{|}>0\big{)}$.
That is, $i\in \mathcal{P}_{2\ell\cdot 2^t +\mathbb{Z}_{2^t}}^{(t)}$ if $i+2^t \in \mathcal{P}_{(2\ell+1)\cdot 2^t + \mathbb{Z}_{2^t}}$.
Hence, the statement holds.

\subsection{Proof of Lemma~\ref{lem:incapable03}}\label{proof:lem:incapable03}

In SC decoding, we have $\alpha_i^{(t)}=f\big{(} \alpha_i^{(t+1)}, \alpha_i^{(t+1)}\big{)}$ since $i_t=0$ for $i\in\mathbb{Z}_{2^t}$.
This implies that any one of $\alpha_i^{(t+1)}$ and $\alpha_{i+2^t}^{(t+1)}$ needs to be zero in order to make $\alpha_i^{(t)}=0$.
Note that both $i$ and $i+2^t$ are in $\mathbb{Z}_{2^{t+1}}$.
Continuing this procedure recursively from stage $t+1$ to stage $n$, only a single variable node with LLR zero at stage $n$ is sufficient to make $\alpha_i^{(t)}=0$.

\subsection{Proof of Lemma~\ref{lem:puncturing01}}\label{proof:lem:puncturing01}

Clearly, $\psi_j^{(0)}=\{\{j\}\}$ by definition,
because an incapable bit is defined as an encoder input bit at decoding stage 0 whose LLR value is zero.
Starting from $\psi_j^{(0)}$,
the family $\psi_j^{(t+1)}$ is directly determined from $\psi_j^{(t)}$, due to the stage-by-stage SC decoding operation.
At decoding stage $t$, for a variable node $v_j^{(t)}$ with $j_t=0$,
we have $\alpha_j^{(t)}=0$ if any one of $\alpha_{j}^{(t+1)}$ and $\alpha_{j+2^t}^{(t+1)}$ is zero.
Therefore, for any $\mathcal{Q}\in\psi_j^{(t)}$ and $\mathcal{B}\in P_{|\mathcal{Q}|}(\{0,2^t\})$, we have $\mathcal{Q}\oplus\mathcal{B} \in \psi_j^{(t+1)}$ if $j_t=0$.
On the other hand, for a variable node $v_j^{(t)}$ with $j_t=1$,
we have $\alpha_j^{(t)}=0$ if both $\alpha_{j}^{(t+1)}$ and $\alpha_{j-2^t}^{(t+1)}$ are zero.
Thus, for any $\mathcal{Q}\in\psi_j^{(t)}$, we have $\mathcal{Q}\cup (-2^t+\mathcal{Q})\in\psi_j^{(t+1)}$ if $j_t=1$.
Taking these two relations between $\psi_j^{(t)}$ and $\psi_j^{(t+1)}$ together, we obtain the recursion formula in \eqref{eq:psi_op}.
Finally, $\psi_j=\psi_j^{(n)}$ is obtained by performing this recursion for $t=[0:n-1]$ with $\psi_j^{(0)}=\{\{j\}\}$.

\subsection{Proof of Theorem~\ref{lem:puncturing02}}\label{proof:lem:puncturing02}

We first show that $\mathcal{D}_j\in\psi_j$.
Let $\mathcal{Q}_{j,0}^{(t)}\in\psi_j^{(t)}$ denote the element set obtained by always setting $\mathcal{B}=\{0,\ldots,0\}$ in the recursive calculation for $j_\ell=0$ in \eqref{eq:psi_op}, Lemma~\ref{lem:puncturing01} for any $\ell\in[0:t]$.
Starting from $\mathcal{Q}_{j,0}^{(0)}=\{j\}$, $\mathcal{Q}_{j,0}^{(n)}$ is then obtained by recursively calculating
\begin{equation}\label{eq:Q_op}
\begin{split}
    \mathcal{Q}_{j,0}^{(t+1)}&=
    \begin{cases}
    \mathcal{Q}_{j,0}^{(t)}, & \text{if } j_t = 0 \\
    \mathcal{Q}_{j,0}^{(t)}\cup \left(-2^t + \mathcal{Q}_{j,0}^{(t)}\right), & \text{if } j_t = 1
    \end{cases}
\end{split}
\end{equation}
for $t\in[0:n-1]$.
This leads to
\begin{equation}\label{eq:Q_eq_D}\nonumber
\begin{split}
    \mathcal{Q}_{j,0}^{(n)}&=\left\{ j - \sum_{t=0}^{n-1} a_t j_t 2^t \mathrel{\Big|} a=\sum_{l=0}^{n-1}a_l 2^l \in \mathbb{Z}_{2^n}\right\} \\
    &= \left\{k\in \mathbb{Z}_{2^n} \mid k \preceq j\right\} \\
    &= \mathcal{D}_j
\end{split}
\end{equation}
where the same elements that appear multiple times are taken into account once.

Now we prove that $\bar{\mathcal{D}}_j\in\psi_j$.
Let $\mathcal{Q}_{j,1}^{(t)}\in\psi_j^{(t)}$ be the element set obtained by setting $\mathcal{B}=\{2^t,\ldots,2^t\}$ in the recursive calculation for $j_\ell=0$ in \eqref{eq:psi_op} for any $\ell\in[0:t]$.
Starting from $\mathcal{Q}_{j,1}^{(0)}=\{j\}$, we obtain $\mathcal{Q}_{j,1}^{(n)}$ by recursively performing
\begin{equation}\label{eq:rQ_op}
\begin{split}
    \mathcal{Q}_{j,1}^{(t+1)}&=
    \begin{cases}
    2^t+\mathcal{Q}_{j,1}^{(t)}, & \text{if } j_t = 0 \\
    \mathcal{Q}_{j,1}^{(t)}\cup \left(-2^t + \mathcal{Q}_{j,1}^{(t)}\right), & \text{if } j_t = 1
    \end{cases}
\end{split}
\end{equation}
for $t\in[0:n-1]$.
The difference between \eqref{eq:Q_op} and \eqref{eq:rQ_op} is whether or not $2^t$ is added at stage $t$, when $j_t=0$.
Therefore, we have
\begin{equation}\label{eq:rQ_and_Q}\nonumber
    \mathcal{Q}_{j,1}^{(n)} = \sum_{t=0}^{n-1} \bar{j}_t 2^t + \mathcal{Q}_{j,0}^{(n)} = \bar{j}+ \mathcal{Q}_{j,0}^{(n)}.
\end{equation}
From \eqref{eq:rQ_and_Q}, we get
\begin{equation}\label{eq:rQ_eq_D}\nonumber
\begin{split}
    \mathcal{Q}_{j,1}^{(n)}&=\bar{j}+\left\{ j - \sum_{t=0}^{n-1} a_t j_t 2^t \mathrel{\Big|} a=\sum_{l=0}^{n-1}a_l 2^l \in \mathbb{Z}_{2^n}\right\} \\
    &= (2^n-1)-\mathcal{D}_j = \bar{\mathcal{D}}_j.
\end{split}
\end{equation}
Here, the second equality comes from the fact that $j+\bar{j} = 2^n-1$.

\subsection{Proof of Theorem~\ref{thm:eq_punct}}\label{proof:thm:eq_punct}

For two integers $i,j\in\mathbb{Z}_{2^n}$ that are not comparable with respect to binary domination,
consider $\psi_{\mathcal{D}_i}$, $\psi_{\mathcal{D}_j}$, and $\psi_{\mathcal{D}_i \cup \mathcal{D}_j}$.
Clearly, $\psi_{\mathcal{D}_i} = \psi_i$ by the transitivity of $\prec$, where $\psi_j$ is given in Lemma~\ref{lem:puncturing01}.
Accordingly, it suffices to consider making both $u_i$ and $u_j$ incapable in order to obtain $\psi_{\mathcal{D}_i \cup \mathcal{D}_j}$.
Since $i$ and $j$ are not comparable, $\psi_j$ has nothing to do with making $u_i$ incapable, and vice versa.
For $\mathcal{Q}_i\in\psi_i$ and $\mathcal{Q}_j\in\psi_j$,
if $\mathcal{Q}_i \cup \mathcal{Q}_j$ has cardinality $\left|\mathcal{D}_i \cup \mathcal{D}_j\right|$, then it becomes a puncturing bit pattern to make both $u_i$ and $u_j$ incapable by Lemma~\ref{lem:incapable01}.
Therefore, we have
\begin{equation}\label{eq:punct_02}
    \psi_{\mathcal{D}_i \cup \mathcal{D}_j}
    = \left\{ \mathcal{A} \in \psi_i \vee \psi_j \mid \left|\mathcal{A} \right|=\left|\mathcal{D}_i \cup \mathcal{D}_j\right|  \right\}.
\end{equation}
Recall that if $j\in\mathcal{U}_p$, then $\mathcal{D}_j \subseteq \mathcal{U}_p$ by Theorem~\ref{thm:incapable01}.
By the transitivity of $\prec$, we have
\begin{equation}\label{eq:punct_03}
    \mathcal{U}_p = \bigcup_{j\in\mathcal{U}_p} \mathcal{D}_j = \bigcup_{j\in\breve{\mathcal{U}}_p} \mathcal{D}_j.
\end{equation}
Applying \eqref{eq:punct_02} to \eqref{eq:punct_03}, we finally have
\begin{equation}\nonumber
    \psi_{\mathcal{U}_p}
    = \psi_{ \bigcup_{j\in\breve{\mathcal{U}}_p} \mathcal{D}_j }
    = \left\{\mathcal{A}\in \bigvee_{j\in\breve{\mathcal{U}}_p} \psi_{j} \mathrel{\Big|}  \left|\mathcal{A}\right| = \left| \mathcal{U}_p \right| \right\}.
\end{equation}

\subsection{Proof of Theorem~\ref{thm:id_punct}}\label{proof:thm:id_punct}

Let $\mathcal{P}^{(t)}$ denote the index set of variable nodes in $\mathcal{V}^{(t)}$ at stage $t$,
whose LLRs are equal to zero.
By feeding the intrinsic LLRs into $\mathcal{V}^{(n)}$, we get $\mathcal{P}^{(n)} = \mathcal{X}_p$.
Consider $j\in\mathbb{Z}_{2^n}$ such that $\alpha_j^{(n)} =0$, that is, $j\in\mathcal{P}^{(n)}$.
By the assumption on $\mathcal{X}_p$, we have
\begin{equation}\label{eq:thm:id_punct_1}\nonumber
\alpha_i^{(n)} = 0,~\forall i \in \hat{\mathcal{D}}_j \subset \mathcal{P}^{(n)}.
\end{equation}
At stage $n-1$ as the first decoding stage,
if $j_{n-1}=0$, we have
\begin{equation}\label{eq:thm:id_punct_2}
    \alpha_j^{(n-1)} = f \left( \alpha_j^{(n)}=0,~\alpha_{j+2^{n-1}}^{(n)} \right) = 0,
\end{equation}
regardless of the value of $\alpha_{j+2^{n-1}}^{(n)}$.
If $j_{n-1}=1$, we have $\alpha_{j-2^{n-1}}^{(n)}=0$ by \eqref{eq:thm:id_punct_1}, so
\begin{equation}\label{eq:thm:id_punct_3}
    \alpha_j^{(n-1)} = g \left(\alpha_{j-2^{n-1}}^{(n)},~\alpha_j^{(n)}=0,~ \beta_{j-2^{n-1}}^{(n-1)}\right) = 0.
\end{equation}
Hence, $\alpha_j^{(n)}=0$ leads to $\alpha_j^{(n-1)}=0$ for any $j\in\mathcal{P}^{(n)}$.
Since it was clearly proven in \cite{Shin2013,Zhang2014} that the numbers of zero LLRs at stages $t$ and $t+1$ are exactly the same for $t\in[0:n-1]$, we have $\mathcal{P}^{(n-1)}=\mathcal{P}^{(n)}$.

Assume that $\mathcal{P}^{(t+1)}=\mathcal{P}^{(n)}$ for $t\in[0:n-2]$.
Then, for any $j\in\mathcal{P}^{(t+1)}$, we have
\begin{equation}\label{eq:thm:id_punct_4}\nonumber
    \alpha_i^{(t+1)} = 0,~\forall i \in \hat{\mathcal{D}}_j \subset \mathcal{P}^{(t+1)}.
\end{equation}
In the same way as in \eqref{eq:thm:id_punct_2} and \eqref{eq:thm:id_punct_3}, we have $\alpha_j^{(t)}=0$.
That is, $\mathcal{P}^{(t)}=\mathcal{P}^{(t+1)}$ for $t\in[0:n-1]$.
Finally, we have $\mathcal{P}^{(0)}=\mathcal{P}^{(n)}=\mathcal{X}_p$ by mathematical induction, and $\mathcal{P}^{(0)}=\mathcal{U}_p$ by definition.

\subsection{Proof of Theorem~\ref{thm:rev_punct}}\label{proof:thm:rev_punct}

Let $\mathcal{P}^{(t)}$ denote the index set of variable nodes with zero LLR in $\mathcal{V}^{(t)}$ at stage $t$ such that $\mathcal{P}^{(n)}=\mathcal{X}_p$.
For an integer $i$ with binary representation $(i_{n-1}i_{n-2}\cdots i_0)$ and a subset $\mathcal{T}\subseteq \mathbb{Z}_{n}$,
we write $\bar{i}_\mathcal{T}$ to denote the partial bitwise complement of $i$ with respect to $\mathcal{T}$, which is obtained by inverting $i_t$ for $t\in\mathcal{T}$.
For example, $\bar{6}_{\{0,2\}} = (\bar{1}1\bar{0})= (011) = 3$.

Consider $j\in\mathbb{Z}_{2^n}$ such that $\alpha_j^{(n)}=0$, that is, $j\in\mathcal{P}^{(n)}$.
Recall that $\hat{\mathcal{G}}_j$ is the dominating integer set of $j$ defined in Definition~\ref{def:dominating_set}.
By the assumption on $\mathcal{X}_p$ with binary domination,
\begin{equation}\label{eq:thm:rev_punct_1}
\alpha_i^{(n)}=0,~\forall i \in\hat{\mathcal{G}}_j \subset \mathcal{P}^{(n)}.
\end{equation}
At stage $n-1$ as the first decoding stage, if $j_{n-1}=0$, we have $\alpha_{j+2^{n-1}}^{(n)}=0$ by \eqref{eq:thm:rev_punct_1}, so
\begin{equation}\label{eq:thm:rev_punct_2}
    \alpha_{j+2^{n-1}}^{(n-1)}=g\left(\alpha_j^{(n)}=0,~\alpha_{j+2^{n-1}}^{(n)},~\beta_j^{(n-1)}\right)=0.
\end{equation}
If $j_{n-1}=1$, we have
\begin{equation}\label{eq:thm:rev_punct_3}
    \alpha_{j-2^{n-1}}^{(n-1)}=f\left( \alpha_{j-2^{n-1}}^{(n)},~\alpha_j^{(n)}=0\right)=0,
\end{equation}
regardless of the value of $\alpha_{j-2^{n-1}}^{(n)}$.
That is, $\alpha_j^{(n)}=0$ leads to $\alpha_{\bar{j}_{\{n-1\}}}^{(n-1)}=0$ for any $j \in \mathcal{P}^{(n)}$.
Thus, the index set of variable nodes with zero LLRs after decoding stage $n-1$ is given by
\begin{equation} \label{eq:thm:rev_punct_4}\nonumber
    \mathcal{P}^{(n-1)}=\left\{\bar{p}_{\{n-1\}} \mid p\in\mathcal{P}^{(n)} \right\}.
\end{equation}

Now, assume that
$\mathcal{P}^{(t+1)}=\left\{ \bar{p}_{\{t+1,\ldots,n-1\}} \mid p\in\mathcal{P}^{(n)} \right\}$
for $t\in[0:n-2]$.
Then, for any $j\in\mathcal{P}^{(t+1)}$, we have
\begin{equation}\label{eq:thm:rev_punct_6}\nonumber
    \left\{ k\in \mathbb{Z}_{2^n} \mid \left(j \prec k\right) \wedge \left( k_r=j_r,\forall r \in [t+1:n-1] \right) \right\}\subset \mathcal{P}^{(t+1)},
\end{equation}
that is,
\if\doccolumn1
\begin{equation}\label{eq:thm:rev_punct_7}\nonumber
    \left\{ k\in \left[\left\lfloor \frac{j}{2^{t+1}} \right\rfloor 2^{t+1}:\left(\left\lfloor \frac{j}{2^{t+1}} \right\rfloor +1\right) 2^{t+1} - 1\right] \mathrel{\Big|} j \prec k  \right\} \subset \mathcal{P}^{(t+1)}.
\end{equation}
\else
\begin{equation}\nonumber
\begin{split}
    \left\{ k\in \left[\left\lfloor \frac{j}{2^{t+1}} \right\rfloor 2^{t+1}:\left(\left\lfloor \frac{j}{2^{t+1}} \right\rfloor +1\right) 2^{t+1} - 1\right] \mathrel{\Big|} j \prec k  \right\}
    \subset \mathcal{P}^{(t+1)}.
\end{split}
\end{equation}
\fi
The index set of variable nodes with LLR zero corresponding to $j$ at decoding stage $t+1$ is also constrained by binary domination within the subset of $2^{t+1}$ variable nodes, $\big{\{}v_k^{(t)}\mid k\in \left[\left\lfloor \tfrac{j}{2^{t+1}} \right\rfloor 2^{t+1}:\left(\left\lfloor \tfrac{j}{2^{t+1}} \right\rfloor +1\right) 2^{t+1} - 1\right]\big{\}}$.
In the same way as in \eqref{eq:thm:rev_punct_2} and \eqref{eq:thm:rev_punct_3},
$\alpha_j^{(t+1)}=0$ results in $\alpha_{\bar{j}_{\{t\}}}^{(t)}=0$ for any $j\in\mathcal{P}^{(t+1)}$,
so we have
\begin{equation}\label{eq:thm:rev_punct_10}\nonumber
\mathcal{P}^{(t)} = \left\{ \bar{p}_{\{t\}} \mid p\in \mathcal{P}^{(t+1)} \right\} = \left\{ \bar{p}_{\{t,t+1,\ldots,n-1\}} \mid p\in \mathcal{P}^{(n)} \right\}
\end{equation}
for $t \in [0:n-2]$. By mathematical induction, we have
\begin{equation}\nonumber
    \mathcal{P}^{(0)} = \left\{ \bar{p}_{\{0,1,\ldots,n-1\}} \mid p\in \mathcal{P}^{(n)} \right\}
    = \bar{\mathcal{P}}^{(n)}.
\end{equation}
Hence, we have $\mathcal{U}_p=\mathcal{P}^{(0)}=\bar{\mathcal{X}}_p$.

\end{document}